\documentclass[a4paper,twocolumn,11pt,accepted=2024-02-28,allowfontchangeintitle]{quantumarticle}
\pdfoutput=1

\usepackage[utf8]{inputenc}
\usepackage[english]{babel}
\usepackage[T1]{fontenc}
\usepackage{amsmath}
\usepackage{hyperref}
\usepackage{epsfig}
\usepackage{amssymb}
\usepackage{amsmath}
\usepackage{amsfonts}
\usepackage{mathtools}
\usepackage{bm}
\usepackage{tikz}
\usepackage{physics}
\usepackage{enumitem}
\usepackage{dsfont}
\usepackage{graphicx,grffile}

\usepackage{algorithm2e}
\RestyleAlgo{ruled}
\definecolor{xblue}{rgb}{0.27,0.45,0.87}
\definecolor{white}{rgb}{0.98,0.98,0.98}
\usepackage[most]{tcolorbox}

\tcbset{colback=white, colframe=xblue }
\usepackage[toc,page]{appendix}
\usepackage[caption=false]{subfig}

\usepackage{amsthm}
\newtheorem{theorem}{Theorem}
\newtheorem{definition}{Definition}
\newtheorem{lemma}[theorem]{Lemma}

\usepackage{cleveref}
\Crefname{equation}{Eq.}{Eqs.}

\renewcommand{\grad}{\text{grad }}

\newcommand{\R}{\mathbb{R}}

\newcommand{\dpar}[2]{\frac{\partial #1}{\partial #2}}
\newcommand{\btheta}{{\bm{\theta}}}
\newcommand{\bx}{{\bm{x}}}

\newcommand{\bphi}{\bm{\phi}}

\newcommand{\GL}[1]{\mathrm{GL}(#1)}
\newcommand{\Ad}[1]{\mathrm{Ad}_{#1}\,}
\newcommand{\ad}[1]{\mathrm{ad}_{#1}\,}
\newcommand{\Span}[1]{\mathrm{span}\left(\left\{#1\right\}\right)}

\newcommand{\SU}[1]{\mathrm{SU}(#1)}
\newcommand{\su}[1]{\mathfrak{su}(#1)}

\newcounter{circuit}[section]

\crefname{table}{table}{tables}
\Crefname{table}{Table}{Tables}
\crefname{circuit}{circuit}{circuits}
\Crefname{circuit}{Circuit}{Circuits}
\Crefname{figure}{Fig.}{Figs.}
\Crefname{figure}{Fig.}{Figs.}

\begin{document}

\title{Here comes the SU(N): multivariate quantum gates and gradients}

\author{Roeland Wiersema}
\affiliation{Vector Institute, MaRS  Centre,  Toronto,  Ontario,  M5G  1M1,  Canada}
\affiliation{Department of Physics and Astronomy, University of Waterloo, Ontario, N2L 3G1, Canada}
\author{Dylan Lewis}
\affiliation{Department of Physics and Astronomy, University College London, London WC1E 6BT, United Kingdom}
\author{David Wierichs}
\affiliation{Xanadu, Toronto, ON, M5G 2C8, Canada}
\author{Juan Carrasquilla}
\affiliation{Vector Institute, MaRS  Centre,  Toronto,  Ontario,  M5G  1M1,  Canada}
\affiliation{Department of Physics and Astronomy, University of Waterloo, Ontario, N2L 3G1, Canada}
\author{Nathan Killoran}
\affiliation{Xanadu, Toronto, ON, M5G 2C8, Canada}


\begin{abstract}
    Variational quantum algorithms use non-convex optimization methods to find the optimal parameters for a parametrized quantum circuit in order to solve a computational problem. The choice of the circuit ansatz, which consists of parameterized gates, is crucial to the success of these algorithms. Here, we propose a gate which fully parameterizes the special unitary group $\SU{N}$. This gate is generated by a sum of non-commuting operators, and we provide a method for calculating its gradient on quantum hardware. In addition, we provide a theorem for the computational complexity of calculating these gradients by using results from Lie algebra theory. In doing so, we further generalize previous parameter-shift methods. We show that the proposed gate and its optimization satisfy the quantum speed limit, resulting in geodesics on the unitary group. Finally, we give numerical evidence to support the feasibility of our approach and show the advantage of our gate over a standard gate decomposition scheme. In doing so, we show that not only the expressibility of an ansatz matters, but also how it's explicitly parameterized.
\end{abstract}
\maketitle

\section{Introduction}

Variational quantum computing is a paradigm of quantum computing that uses optimization algorithms to find the optimal parameters for a parameterized quantum circuit~\cite{Cerezo2021vqa,Tilly2022vqe}. Crucial for the success of such algorithms is the choice of circuit ansatz, which usually consists of multiple parameterized one and two-qubit gates. Typically, these gates are parameterized unitary matrices generated by single Pauli-string operators that can locally rotate a state around some axis: $U(t) = \exp{i t G}$, where $t$ is a gate parameter and $G$ a Pauli string. For a specific family of cost functions, there exist a variety of methods that allow one to obtain the gradient with respect to $t$~\cite{ Li2017grads,Mitarai2018grads, Schuld2019grads,Crooks2019grads, Izmaylov2021grads, Wierichs2022grads, Oleksandr2021grads} on quantum hardware. With these gradients, the cost function can be minimized via any gradient-descent-based algorithm.

Instead of considering a gate generated by a single Pauli string, one can construct more general parameterized gates that can perform an arbitrary rotation in $\SU{N}$, the special unitary group. These general $\SU{N}$ rotations are used in a variety of quantum algorithms~\cite{Theis2023propershiftrules, Slattery2022uniblock,Liu2019fewer, Abhinav2017hardware}. In practice, rotations in $\SU{N}$ can be implemented by composing several simple parameterized gates together into a more complicated one. For example, for single and two-qubit gates (where $N=2,4$, respectively), there exist several general decomposition schemes of such gates into products of single-qubit gates and CNOTs~\cite{ Khaneja2001cartan,kraus2001optimal,Vatan2004opt2qubit, vatan2004realization,Vartiainen2004, Dalessandro2006decompuni}. In practice, this compilation comes with hardware-specific challenges, since quantum hardware usually has a set of native gates into which all others have to be decomposed~\cite{Zulehner2019compiling, Foxen2020compiling}. 

Choosing the right parameterization for a function is important because it can significantly affect the properties of its gradients. Reparameterizing functions to obtain more useful gradients is a well-known method in statistics and machine learning. For example, in restricted maximum likelihood methods one can ensure numerical stability of quasi-Newton methods by decomposing covariance matrices into Cholesky factors~\cite{groeneveld1994reparameterization}. In addition, methods like auxiliary linear transformations~\cite{raiko2012deep}, batch normalization~\cite{ioffe2015batch} and weight normalization~\cite{salimans2016weight} are used to improve the gradients in neural networks. In variational inference, the reparameterization trick~\cite{price1958useful} is at the core of variational autoencoder approaches and allows for gradients for stochastic back-propagation~\cite{rezende2014stochastic, welling2014vae}. In light of this, it may be worthwhile to investigate alternative parameterizations of quantum gates for variational quantum algorithms.

In this work, we propose a family of parameterized unitaries called $\mathbb{SU}(N)$ gates and provide a method to evaluate their gradients on quantum hardware. In doing so, we generalize the prior literature one step further, since many past schemes can be understood as special cases of our proposal~\cite{ Li2017grads,Mitarai2018grads, Schuld2019grads,Crooks2019grads, Izmaylov2021grads, Wierichs2022grads, Oleksandr2021grads}. We provide numerical results to support the validity of our approach and give several examples to illustrate the capabilities of the $\mathbb{SU}(N)$ gate. We show that this gate satisfies the quantum speed limit and that it is easier to optimize compared to $\SU{N}$ parameterizations that consist of products of gates. We argue that this is the case because the product of unitaries creates a ``bias'' in the Lie algebra that deforms the cost landscape. In addition, we highlight the connections between our formalism and the properties of semisimple Lie algebras and establish a bound on the computational complexity of the gradient estimation using tools from representation theory.
\section{\texorpdfstring{$\mathbb{SU}(N)$ gates}{sun gates} \label{sec:method}}

A quantum gate is a unitary operation $U$ that acts on a quantum state $\rho$ in a complex Hilbert space. If we ignore a global phase, then a gate $U$ acting on $N_{\mathrm{qubits}}$ qubits is an element of the special unitary group $\SU{N}$ (see App. \ref{app:sun}), where $N = 2^{N_{\mathrm{qubits}}}$. Note that all of the following works for any $N>1$, but here we restrict ourselves to the qubit case. We are interested in constructing a quantum gate that parameterizes all of $\SU{N}$. To achieve this, we make use of the theory of Lie algebras. We will not be concerned with the formal treatment of this topic, which can be found in many excellent textbooks~\cite{hall2013lie, fulton2013representation, rossmann2002lie}. 

To construct our gate, we realize that $\SU{N}$ is a (semisimple) Lie group and so there exists a unique connection between its elements and the Lie algebra $\su{N}$ via the so-called Lie correspondence, or Lie's third theorem~\cite{Serre2009lie, rossmann2002lie}. In particular, each $g \in \SU{N}$ can be identified with an $A \in \su{N}$ via the exponential map $g = \exp{A}$. For our purposes, we can understand the Lie algebra $\su{N}$ as a vector space of dimension $N^2-1$ that is closed under the commutator, $\comm{A}{B} = AB - BA\in\su{N}$ for $A,B \in \su{N}$.
For $\su{N}$, we choose as a basis the tensor products of Pauli matrices multiplied by the imaginary unit $i$:
\begin{align}
    \mathcal{P}^{(N_{\mathrm{qubits}})} =\left\{ i(\sigma_1 \otimes \ldots\otimes \sigma_{N_{\mathrm{qubits}}})\right\}\setminus \{I_{N_{\mathrm{qubits}}}\},\label{eq:paulis}
\end{align}
where $ \sigma_i \in \{I,X,Y,Z\} $ and $I_{N_{\mathrm{qubits}}}=iI^{\otimes N_{\mathrm{qubits}}}$.

We choose the following parameterization of $\SU{N}$:

\begin{align}
    U(\btheta) = \exp{ A(\btheta)},\quad  A(\btheta) = \sum_m \theta_m G_m,\label{eq:gate}
\end{align}
\noindent
where $\btheta = (\theta_1,\theta_2,\ldots, \theta_{N^2-1}) \in \mathbb{R}^{N^2-1}$ and $\{G_m\} \in \mathcal{P}^{(N_{\mathrm{qubits}})}$. To distinguish between the group and the gate, we call the parameterization in \Cref{eq:gate} a $\mathbb{SU}(N)$ gate. The coordinates $\btheta$ are called the canonical coordinates, which uniquely parameterize $U$ through the Lie algebra $\su{N}$. Since we typically cannot implement the above gate in hardware, we will have to be decompose it via a standard unitary decomposition algorithm~\cite{ Khaneja2001cartan,kraus2001optimal,Vatan2004opt2qubit, vatan2004realization,Vartiainen2004, Dalessandro2006decompuni}. We emphasize here that even though the gate will be decomposed, it is parameterized as an exponential map. Hence we can understand \Cref{eq:gate} as a change of coordinates from the $\SU{N}$ gate decomposition.

If we do not want to parameterize all of $\SU{N}$, we can instead parameterize a more restricted Hamiltonian by setting some of the parameters $\theta_m$ to zero. This makes \Cref{eq:gate} a natural parameterization of several Hamiltonians available on modern quantum hardware platforms. These Hamiltonians often have multiple independently tunable fields which can be active at the same time and do not necessarily commute. One typically has local control on each qubit and access to an interacting Hamiltonian between pairs of qubits, depending on the topology of the quantum device~\cite{Schuch2003gates}. The interacting pair can for example be a $ZZ$ interaction for Josephson flux qubits~\cite{Orlando1999supercond}, a Heisenberg interaction for nuclear spins in doped silicon~\cite{Kane1998nuc} or an $XY$ interaction in quantum dots interacting with a cavity~\cite{Imamoglu1999qcavity}. 

To use this gate in a gradient-based variational quantum algorithm, we have to be able to obtain partial derivatives of $U(\btheta)$ with respect to each parameter $\theta_l$. Although there exist a variety of works that provide analytical expressions for gradients through quantum circuits via the parameter-shift rule~\cite{Li2017grads,Mitarai2018grads, Schuld2019grads,Crooks2019grads, Izmaylov2021grads, Wierichs2022grads, Oleksandr2021grads, leng2022differentiable}, these works almost uniformly assume that the gate is of the form $U(\theta) =\exp{i\theta P}$, where $P$ is a Hermitian operator. As far as we are aware, the only methods to obtain gradients of \Cref{eq:gate} with respect to $\btheta$ are the stochastic and Nyquist parameter-shift rules of~\cite{Crooks2019grads} and~\cite{Theis2023propershiftrules}, respectively. The first approach relies on an integral identity for bounded operators that is estimated via Monte Carlo~\cite{Wilcox1967exponential}, whereas the latter is based on a theorem in Fourier analysis~\cite{whittaker1915}.

\section{Obtaining the gradient} \label{sec:grad}
Here, we provide a new approach to obtain the gradient of \Cref{eq:gate} that makes use of differentiable programming, which is efficient for gates acting on a small number of qubits. To start, we note that the partial derivative with respect to a parameter $\theta_l$ is given by
\begin{align}
    \dpar{}{\theta_l} U(\btheta) &= \dpar{}{\theta_l}\exp{A(\btheta)}\nonumber\\
    &= U(\btheta) \sum_{p=0}^\infty \frac{(-1)^p}{(p+1)!}(\mathrm{ad}_{A(\btheta)})^p\dpar{}{\theta_l}A(\btheta).\label{eq:adj}
\end{align}
Here, $\mathrm{ad}_X$ denotes the adjoint action of the Lie algebra given by the commutator $\ad{X}(Y) = \comm{X}{Y}$~\cite{rossmann2002lie}. Furthermore, we write $(\ad{X})^p(Y) = \comm{X}{\comm{X}{\ldots \comm{X}{Y}}}$, hence $(\ad{X})^p$ denotes a nested commutator of $p$ terms. For more details, see App. \ref{app:BCH}. Note that the term on the right of $U(\btheta)$ in \Cref{eq:adj} is an element of the Lie algebra, since $\partial / \partial \theta_l A(\btheta)=G_l \in \su{N}$ and so the commutator keeps the entire sum in the algebra. For notational clarity we define
\begin{align}
    \Omega_l(\btheta) =\sum_{p=0}^\infty \frac{(-1)^p}{(p+1)!}(\mathrm{ad}_{A(\btheta)})^p\dpar{}{\theta_l}A(\btheta),\label{eq:u_dtheta}
\end{align}
where $\Omega_l(\btheta)\in\su{N}$ is a skew-Hermitian operator that generates a unitary, which we call the \emph{effective generator}. Given that \Cref{eq:u_dtheta} is an infinite series of nested commutators it is not clear how $\Omega_l(\btheta)\in\su{N}$ can be calculated in practice without truncating the sum.

We can think of $U$ as a function $U:\mathbb{R}^{N^2-1}\to \SU{N}$ that we evaluate at the point $\btheta$. Since $\SU{N}$ is a differentiable manifold, we can define a set of local coordinates on the group and represent $U(\bx)$ as a matrix described by $N^2-1$ real numbers. Hence, we can think of our gate as a coordinate transformation between the parameters $\bx$ and the entries of the matrix representing the unitary. Since $U(\bx)$ depends smoothly on $x_l$ via the matrix exponential, this coordinate transformation comes with a corresponding Jacobian (or more accurately, pushforward) $dU(\bx): T_{\bx} \mathbb{R}^{N^2-1}\to T_{U(\bx)} \SU{N}$ that maps vectors tangential to $\mathbb{R}^{N^2-1}$ to vectors tangential to $\SU{N}$. We can obtain this Jacobian by differentiating the elements $U_{nm}(\bx)$ with respect to $x_l$:
\begin{align}
    \dpar{}{x_l} U_{nm}(\bx) =
    \partial_{x_l} \mathfrak{Re}[U_{nm}(\bx)] + i \partial_{x_l}\mathfrak{Im}[U_{nm}(\bx)] \label{eq:mat_der}.
\end{align}
To obtain the above matrix function numerically, we rely on the fact that the matrix exponential and its derivative are implemented in differentiable programming frameworks such as JAX~\cite{Jax2018github}, PyTorch~\cite{paszke2019pytorch} and Tensorflow~\cite{tensorflow2015} through automatic differentiation. Here we make use of the JAX implementation, which provides the matrix exponential through a differentiable Pad\'e approximation~\cite{expm, al2010new}. 

Continuing, we note that evaluating $\partial U(\bx)/ \partial x_l $ at a point $\btheta$ produces an element of the tangent space $T_{U(\btheta)} \SU{N}$. We can move from the tangent space to the Lie algebra by left multiplying the elementwise derivative of \Cref{eq:mat_der} in \Cref{eq:adj} with $U^\dag(\btheta)$ (see App. \ref{app:sun}),
\begin{align}
   U^\dag(\btheta) \left(\dpar{}{x_l} U(\bx)\bigg|_{\btheta} \right) =  U^\dag(\btheta)U(\btheta)\Omega_l(\btheta) = \Omega_l(\btheta)\label{eq:omega},
\end{align}
which allows us to obtain $\Omega_l(\btheta)$ exactly, up to machine precision. We emphasize that these steps can be performed on a classical computer, with a cost that is only dependent on the number of qubits the gate acts on, not the number of qubits in the circuit.

We now make the following observation: $\Omega_l(\btheta)$ corresponds to a tangent vector on $\SU{N}$ and generates the one-parameter subgroup $V(t) = \exp{t\Omega_l(\btheta)}$ such that
\begin{align}
    \Omega_l(\btheta) = \frac{d}{dt}\exp{t\Omega_l(\btheta)} \big |_{t=0}
\end{align}
and
\begin{align}
    \dpar{}{\theta_l} U(\btheta) = U(\btheta)\frac{d}{dt}\exp{t\Omega_l(\btheta)} \big |_{t=0}. \label{eq:dU_omega}
\end{align}
We sketch this procedure schematically in \Cref{fig:one_parameter}.
\begin{figure}[htb!]
    \centering
    \includegraphics[width=0.9\columnwidth]{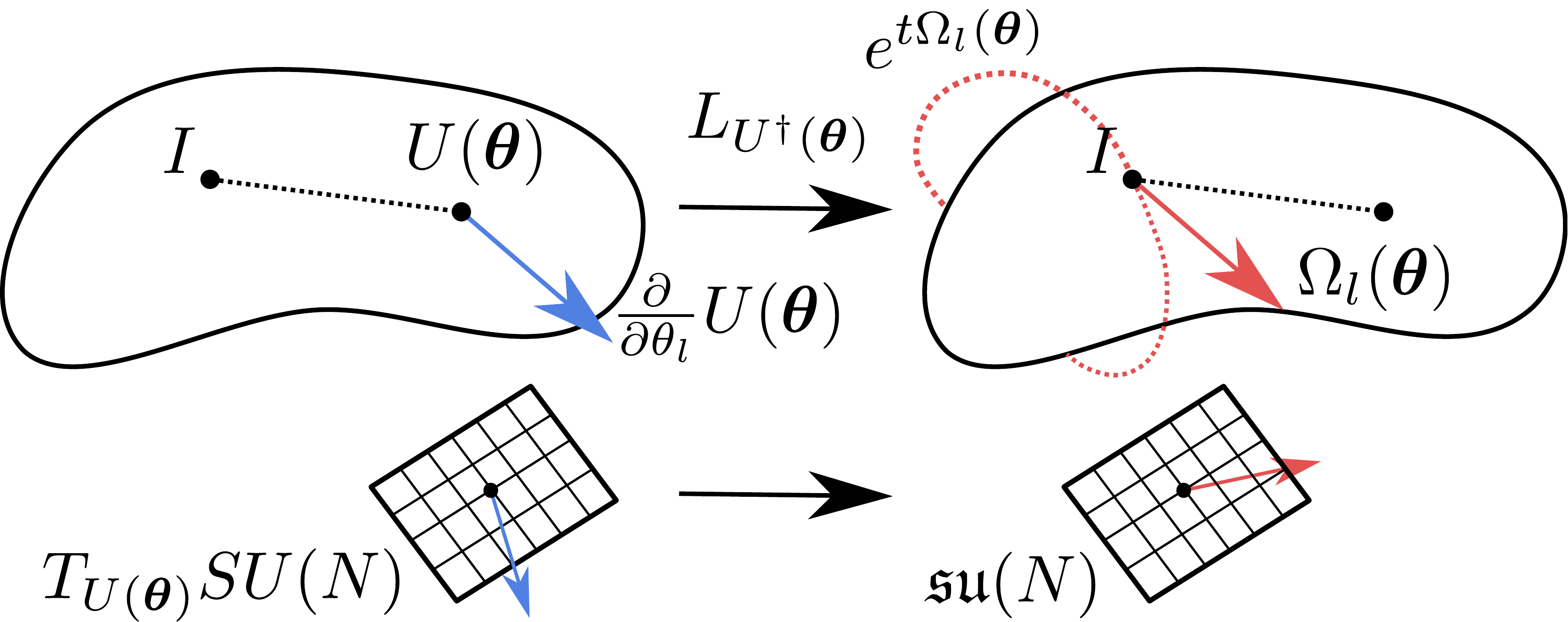}
    \caption{Schematic depiction of our approach. We move to the Lie algebra from the tangent space by left multiplication with  $U^\dag(\btheta)$ and obtain $\Omega_l(\btheta)$. The orbit generated by $\Omega_l(\btheta)$ corresponds to the gate we have to insert in the circuit to compute the gradient.}
    \label{fig:one_parameter}
\end{figure}

We now consider a typical variational setting, where we are interested in minimizing the following cost function:
\begin{align}
    C(\btheta) = \Tr{U(\btheta)\rho U^\dag(\btheta) H},\label{eq:cost}
\end{align}
where $H$ is some Hermitian operator and $\rho$ the initial state of the system. For simplicity we consider a circuit consisting of a single $\mathbb{SU}(N)$ gate. Differentiating the cost function with respect to $\theta_l$ gives
\begin{align}
    \dpar{}{\theta_l} C(\btheta) &= \Tr{\left(\dpar{}{\theta_l}U(\btheta)\right)\rho U^\dag(\btheta)H}+\mathrm{h.c.}
\end{align}
Then, plugging in \Cref{eq:dU_omega} we find,

\begin{align}
    \dpar{}{\theta_l} & C(\btheta)=\nonumber\\
    &\frac{d}{dt}\Tr{\left(U(\btheta)  e^{t\Omega_l(\btheta)}\rho e^{-t\Omega_l(\btheta)}U^\dag(\btheta)\right)H  } \bigg\vert_{t=0}\label{eq:circuit_with_omega},
\end{align}
\noindent
where we used the skew-Hermitian property of the tangent vector $\Omega_l^\dag(\btheta) = -\Omega_l(\btheta)$. Note that \Cref{eq:circuit_with_omega} corresponds to a new circuit with the gate $\exp{t\Omega_l(\btheta)}$ inserted before $U(\btheta)$ (see \Cref{fig:circuit_omega}). 
\begin{figure}[htb!]
    \centering
    \includegraphics[width=\columnwidth]{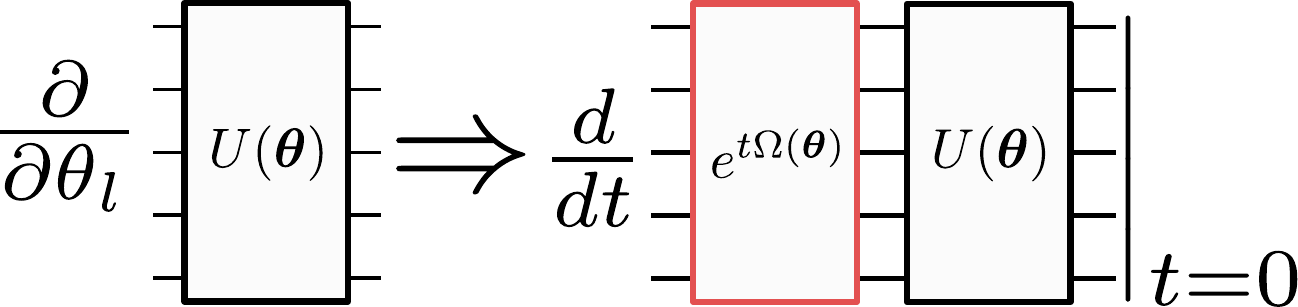}
    \caption{The partial derivative with respect to the gate parameter $\theta_l$ can be obtained by adding a gate to the circuit that is generated by $\Omega_l(\btheta)$. Calculating the derivative with respect to $t$ and evaluating at $t=0$ then provides one with the correct gradient.}
    \label{fig:circuit_omega}
\end{figure}

The gradient of this new circuit can be computed on quantum hardware with a generalized parameter-shift rule (GPSR)~\cite{Wierichs2022grads, Izmaylov2021grads, Oleksandr2021grads}. In \Cref{alg:grad}, we outline the entire algorithm for our gradient estimation and we denote the GPSR subroutine with \texttt{gpsr}. An alternative to the generalized shift rule is to decompose the effective generators and apply the original two-term parameter-shift rule to the constituents (see App.~\ref{app:decompose_generators} for details).
In~\cite{Banchi2021measuringanalytic}, the authors proposed the so-called stochastic parameter-shift rule for multivariate gates, which is based on the Monte Carlo approximation of an operator identity.

In \Cref{fig:exact_grad} we consider a toy example using a random Hamiltonian on a single qubit and compare the exact derivative of an $\mathbb{SU}(2)$ gate with our generalized parameter-shift method (\Cref{alg:grad}), the stochastic parameter-shift rule and the central finite difference derivative with shifts $\pm\frac{\delta}2$. In particular, we consider the gate $U(\btheta)=\exp(ia X + ib Y)$ with $\btheta = (a,b)$ and compute the partial derivative with respect to $a$ over the range $a\in[0,\pi]$ for three fixed values of $b$ on a state vector simulator (without shot noise).
For the finite difference recipe we use $\delta=0.75$, which we found to be a reasonable choice for a shot budget of $100$ shots per cost function evaluation (see App. \ref{app:fd}).
We observe that the generalized $\SU{N}$ derivative reproduces the exact value while the finite difference derivative is slightly biased. This is to be expected because the latter is an approximate method.
While decreasing the shift size $\delta$ reduces the deterministic approximation error, it leads to larger overall estimation errors in shot-based computations like on quantum computers (see App. \ref{app:fd} and e.g.,~\cite{bittel2022fast}).
Finally, the stochastic parameter-shift rule yields an unbiased estimator for the exact derivative but has a finite variance, which we estimated using $100$ samples (see App. \ref{app:sps}). We stress that this variance is a property of the differentiation method itself and not due to sampling on the quantum computer.
All methods require two unique circuits per derivative but the stochastic shift rule needs additional circuits in order to suppress the variance. We provide the code for all our numerical experiments at~\cite{our_code}.

\begin{figure}[htb!]
    \centering
    \includegraphics[width=\columnwidth]{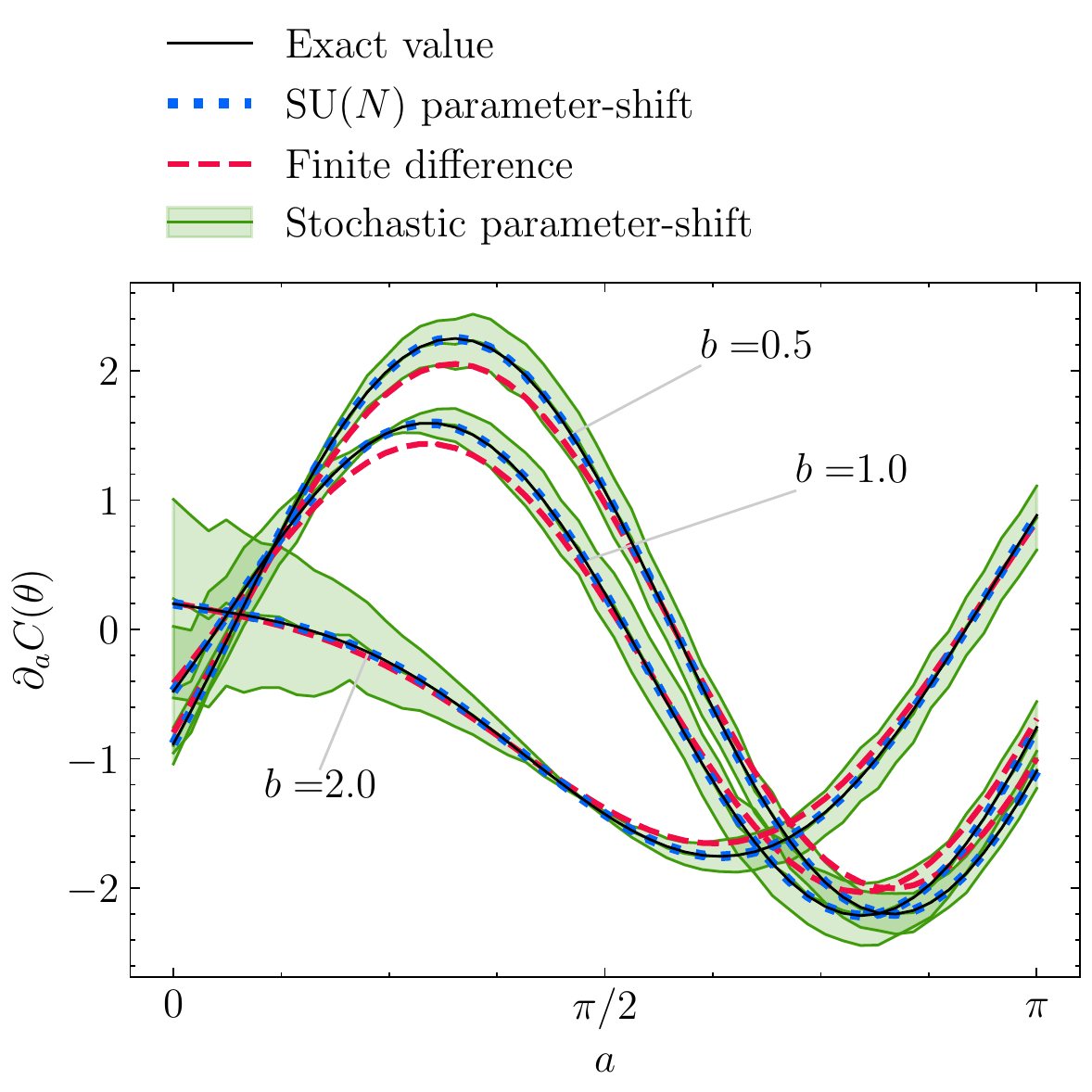}
    \caption{
    Gradients of $C(\btheta)$ for a single $\mathbb{SU}(2)$ gate and a random single-qubit Hamiltonian, in the limit of infinitely many shots on quantum hardware.
    We take $A(\btheta) = ia X + ib Y$ where $\btheta =(a,b)$ and consider the fixed values $b=0.5,1.0,2.0$ together with $a\in [0, \pi]$.
    Our generalized shift rule (dotted) reproduces the exact value (solid), whereas the central finite difference (dashed) is biased and the stochastic shift rule (solid, shaded) comes with a finite statistical error even without shot noise from the quantum measurements.
    Since we look at a single-qubit operation, $\Omega_a(\btheta)$ has a single spectral gap, so we require two shifted circuits to calculate the gradient entry (see App. \ref{app:gps} for details). The finite difference and the stochastic shift rule require two circuits as well, but additional executions are need for the latter to reduce the shown single-sample error.
    }
    \label{fig:exact_grad}
\end{figure}

In addition, we compare the three methods in the presence of shot noise in \Cref{fig:sampled_grad}. We show the means and \emph{single-shot} errors estimated with $1000$ shots, which we split over $100$ samples for the stochastic shift rule. We observe that the generalized $\SU{N}$ shift rule systematically performs best. It is not only unbiased but also has the smallest variance. 
Note that for smaller parameters $b$, the $\SU{N}$ shift rule and the stochastic shift rule show very similar variances.
This is because $U(\btheta)$ approaches the gate $R_X(a)=\exp(iaX)$, which can be differentiated with the original parameter-shift rule, and both rules indeed reduce to the two-term shift rule for $R_X$.

\begin{figure}[htb!]
    \centering
    \includegraphics[width=\columnwidth]{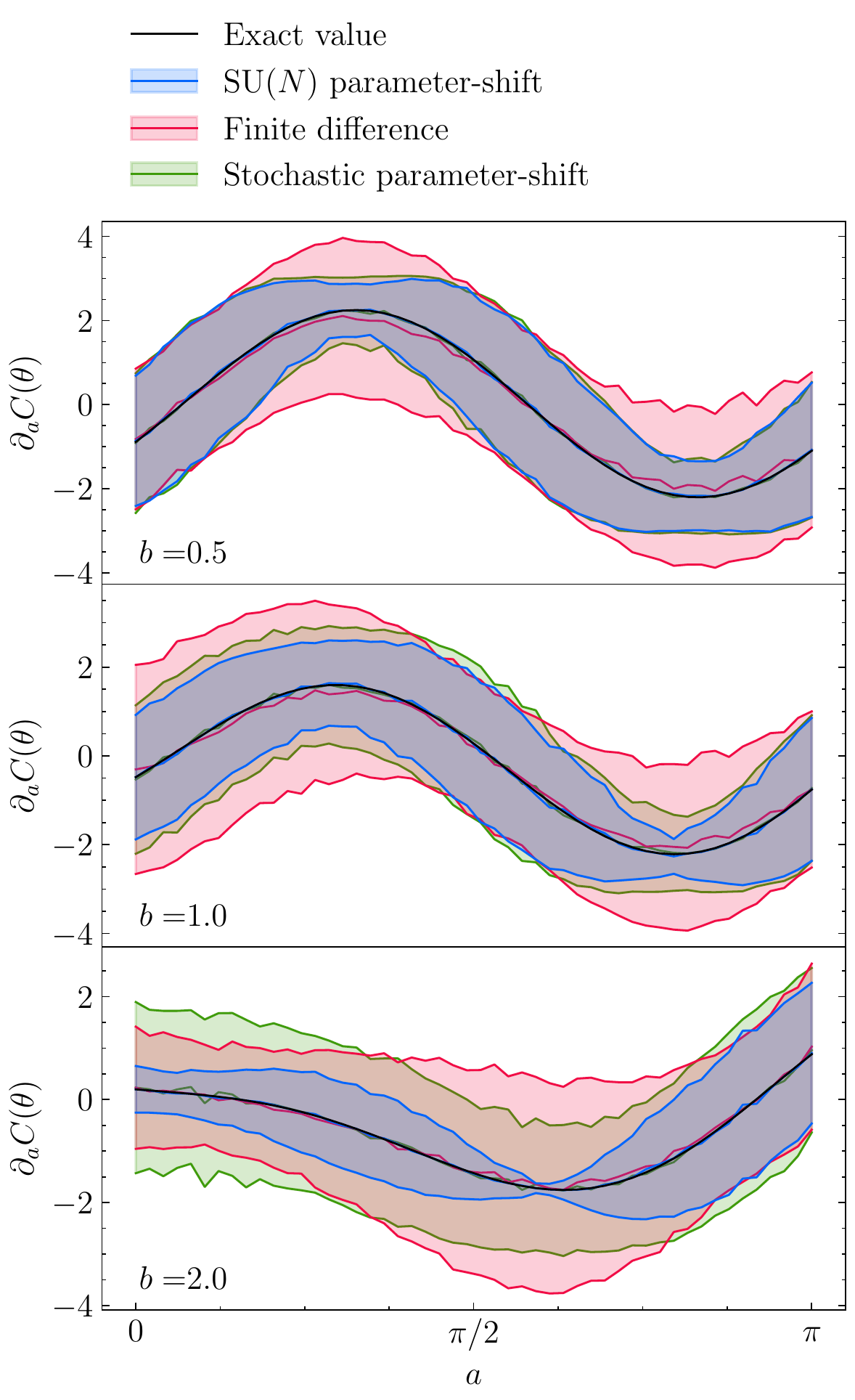}
    \caption{
    Gradients of $C(\btheta)$ as in \Cref{fig:exact_grad} but for finitely many shots on quantum hardware.
    We show the single-shot error for each method, estimated with $1000$ shots, which varies with the gate parameters as noted e.g.,~in~\cite{Oleksandr2021grads}.
    Our generalized $\SU{N}$ shift rule systematically outperforms the other methods.
    For small $b$, the $\SU{N}$ and the stochastic shift rule approach the single-parameter shift rule and hence behave similarly.
    The finite difference shift $\delta=0.75$ is chosen such that the bias and variance are traded off reasonably for $100$ shots (see App. \ref{app:sps} and e.g.,~\cite{bittel2022fast}).
    For other shot numbers, $\delta$ needs to be optimized anew, whereas the parameter-shift rules are known to perform optimally at fixed shifts.
    }
    \label{fig:sampled_grad}
\end{figure}

Finally, we note that \Cref{eq:circuit_with_omega} is closely related to the Riemannian gradient on $\SU{N}$ ~\cite{SchulteHerbruggen2010gradflow, Wiersema2022riemann}. However, instead of a gradient flow on a Lie group, we have defined a flow on the Lie algebra $\su{N}$, which we retract back to the manifold via the exponential map. This subtle difference induces a different flow from the $\SU{N}$ one, as we illustrate in App. \ref{app:riemann}.


\begin{algorithm}
\caption{$\SU{N}$ gradients.}\label{alg:grad}
\KwIn{$ U(\bx)$, $\rho$, $H$, $\btheta$}
\text{Obtain the Jacobian function:}\\
\For{$l \in (1,\ldots, N^2-1$)}{
$\partial_{x_l}U_l(\bx) = \partial_{x_l} \mathfrak{Re}[U(\bx)] + i \partial_{x_l}\mathfrak{Im}[U(\bx)]$
}
\text{For each gradient step:}\\
\For{$l \in (1,\ldots, N^2-1$)}{
    $\Omega_l(\btheta) \gets U^\dag(\btheta) dU_l(\bx)|_{\btheta}$\\
    $C(t) \gets \Tr{ U(\btheta)e^{t\Omega_l(\btheta)}\rho e^{-t\Omega_l(\btheta)} U^\dag(\btheta) H }$\\
    $\dpar{}{\theta_l} C(\btheta) \gets \texttt{gpsr}(\Omega_l(\btheta))$
}
\end{algorithm}

\section{\label{sec:speed_limits}Comparison with decomposed unitaries}
Previous parameterizations of $\SU{N}$ unitaries consist of products of single-qubit gates and CNOTs~\cite{ Khaneja2001cartan,kraus2001optimal,Vatan2004opt2qubit, vatan2004realization,Vartiainen2004, Dalessandro2006decompuni}. We refer to this parameterization as \emph{decomposed} $\mathbb{SU}(N)$ gates. On the other hand, \Cref{eq:gate} describes a general $\SU{N}$ unitary by exponentiating a parameterization of the Lie algebra $\su{N}$. Here, we investigate the effects of this alternative parameterization.

\subsection{\label{sec:gate_speed_limits}Gate speed limit}
First, we investigate a speed limit in terms of the gate time. 
We slightly modify the definition of~\Cref{eq:gate} for a unitary evolution of the system, $U(\btheta;t) \in \SU{N}$, to include a time $t\in \mathbb{R}^+$, 
\begin{align}
\label{eq:unitary_evolution_with_t}
    U(\btheta;t) = \exp{\bar{A}(\btheta)t},
\end{align}
where $\bar{A}(\btheta) = A(\btheta)/\sqrt{\Tr{A(\btheta)^\dagger A(\btheta)}}$ is a normalized time-independent Hamiltonian (the imaginary unit $i$ is included in $A(\btheta)$). 
The normalization of $\bar{A}(\btheta)$ is equivalent to the normalization of $\btheta$ in Euclidean norm, see Lemma~\ref{lemma:trace_of_square} in App.~\ref{app:gate_speed_limit}.
The normalization of the Hamiltonian (or, equivalently, $\btheta$) means that the total path length covered by the evolution is directly proportional to the evolution time $t$, since we are effectively setting the speed of the evolution to $1$.

The Lie group $\SU{N}$ can be turned into a Riemannian manifold by equipping it with the Hilbert-Schmidt inner product 
$g(x,y) = \Tr{x^\dagger y}$. 
The unitary evolution $U(\btheta;t)$, parameterized by $t$, is a one-parameter subgroup that gives the geodesic~\cite[Theorem III.6]{SchulteHerbruggen2010gradflow} from the identity element at time $t=0$. Geodesics can be defined as generalizations of straight lines in Euclidean geometry.
Using Lemma~\ref{lemma:path_distance} (App.~\ref{app:gate_speed_limit}), the length of the path after time $t$ is constant for time-independent normalized Hamiltonians with $\vert\btheta\vert = 1$, 
\begin{align}
    L[U(\btheta;t),t] =  \sqrt{N} t.
\end{align}
In general, there is more than one geodesic between two points on the manifold. For example, two points on the Bloch sphere can be connected by rotations about the same axis moving in opposite directions. Using Lemma~\ref{lemma:minimal_path} (App.~\ref{app:gate_speed_limit}), one of these geodesics must be the curve of the minimal path length. Hence, the minimum time to generate the evolution $U(\btheta;t_g)$ is $t_g$ along the geodesic of the minimal path.  
For an initial state $\rho$ and final state $\rho_f$, the Fubini-Study metric is used to find a minimum evolution time 
\begin{align}
    t_g = \frac{1}{\sqrt{N}}\arccos( \sqrt{\Tr{\rho \rho_f}}),
\end{align}
giving the Mandelstam-Tamm bound for time-independent normalized Hamiltonians.

In practice, we may only have access to a restricted family of gates within $\SU{N}$, for example due to hardware limitations, in which case we require a decomposition of a desired gate in $\SU{N}$ into gates from this family. Here we want to compute the additional evolution time required by such a decomposition. The simplest gate decomposition 
is to break the unitary into two terms, $U(\btheta;t_g) = U(\bphi^{(2)};t_2)U(\bphi^{(1)};t_1)$.
The parameters $\bphi^{(1)}$ and $\bphi^{(2)}$ are also normalized Hamiltonians, i.e., they have the norm $\vert \bphi^{(1)} \vert =\vert \bphi^{(2)} \vert = 1$. The following theorem shows that using a decomposed circuit over an $\mathbb{SU}(N)$ gate gives an additional evolution time, which corresponds to longer circuit run times.

\begin{theorem}
\label{th:gate_speed_limit}
	For unitary gates generated by normalized time-independent Hamiltonians, consider a general circuit decomposition of two gates $U(\bphi^{(2)};t_2)U(\bphi^{(1)};t_1)$.  There exists an equivalent evolution with an $\mathbb{SU}(N)$ gate $U(\btheta;t_g) =U(\bphi^{(2)};t_2)U(\bphi^{(1)};t_1)$, with evolution time $t_g$, such that 
	\begin{align*}
		t_g \leq t_1+t_2,
	\end{align*}
    with equality if $\bphi^{(1)} +\bphi^{(2)} =\btheta$.
\end{theorem}
The proof of the theorem is in App.~\ref{app:gate_speed_limit}. As expected, a decomposition into two gates gives a longer total evolution time than is possible with an $\mathbb{SU}(N)$ gate due to the normalizations of $\bphi^{(1)}$,  $\bphi^{(2)}$, and $\btheta$. 
A decomposition into more gates would generally lead to an even greater evolution time. A corollary of Theorem~\ref{th:gate_speed_limit} is that any circuit with multiple non-commuting layers of gates cannot be optimal in total time.

\subsection{Unbiased cost landscapes}

An additional advantage of the $\mathbb{SU}(N)$ gate is that it weighs all optimization directions equally. In contrast, a parameterization of $\SU{N}$ in terms of a product of gates will create a bias in the parameter space. We illustrate this point with the following example.
Consider the decomposed $\mathbb{SU}(2)$ gate $V(\btheta) = R_Z(\theta_3)R_Y(\theta_2)R_Z(\theta_1)$ where $R_A(\theta) = \exp{i\theta A}$ and $A=X,Y,Z$. This is the ZYZ decomposition. Using similar techniques as in App. \ref{app:gate_speed_limit}, we can rewrite $V(\btheta)$ to be parameterized in terms of the Lie algebra:
\begin{align}
    V(\btheta) = \exp{i \bphi \cdot \bm{\sigma}},
\end{align}
where $\bm{\sigma} = (X, Y, Z)$ and
\begin{align}
    \bphi = \frac{\arccos(\cos(\theta_2)\cos(\theta_1 + \theta_3))}{\sqrt{1-\cos^2(\theta_2)\cos^2(\theta_1 + \theta_3)}}\nonumber\\
    \times
    \begin{pmatrix}
        \sin(\theta_2)\sin(\theta_1 - \theta_3) \\
        \sin(\theta_2)\cos(\theta_1 - \theta_3) \\
        \cos(\theta_2)\sin(\theta_1 + \theta_3) 
    \end{pmatrix}.
\end{align}
If we look at the components of $\bphi$, we see that the different directions in the Lie algebra are stretched or compressed as a result of the particular choice of parameterization. Consider the normalization $\vert\theta_1\vert + \vert\theta_2\vert + \vert\theta_3\vert = 1$ for the ZYZ decomposition and $|\btheta|=1$ for the $\mathbb{SU}(N)$ gate. With each Hamiltonian term normalized to 1, the prefactor gives the evolution time. 
These choices of norm give equal total evolution times for the ZYZ decomposition and $\mathbb{SU}(2)$ gate, $T_{\textrm{ZYZ}} = T_{\SU{N}} = \sqrt{2}$, irrespective of the specific parameters chosen. In \Cref{fig:bias_comparison_ZYZ_SU2}, we graphically illustrate the Lie algebra deformation by showing the $\bphi$ surface for both the ZYZ decomposition and $\mathbb{SU}(2)$ gate. Note that we have not considered any cost function here; the bias occurs at the level of the parameterization of a specific unitary.
\begin{figure}[htb!]
    \centering
    \includegraphics[scale=0.6]{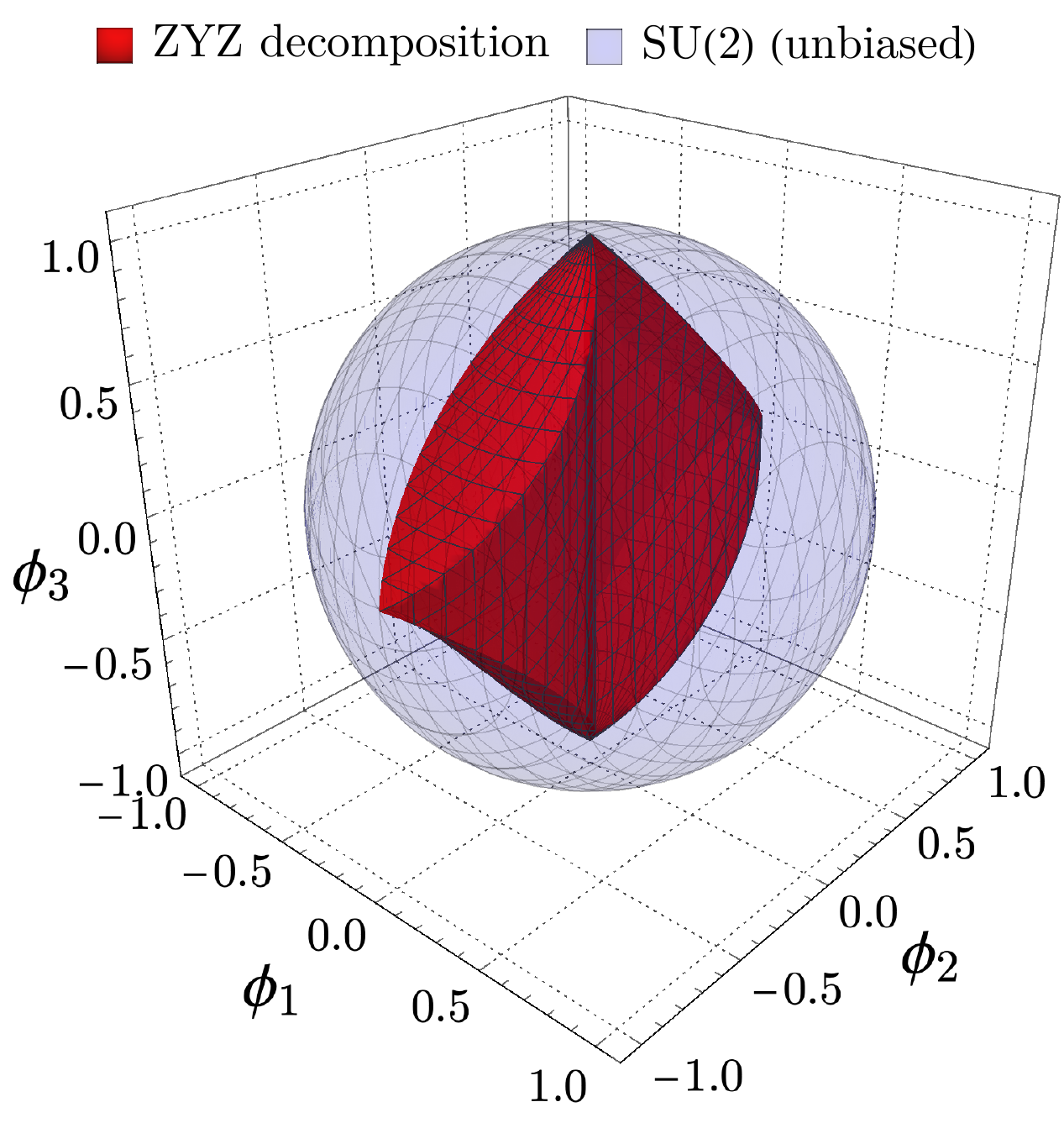}
    \caption{The total unitary evolution for the ZYZ decomposition (red) and the $\mathbb{SU}(2)$ gate (blue) can be expressed in the form $\exp{i \bm{\phi} \cdot \bm{\sigma}}$. The components $\bm{\phi} = (\phi_1, \phi_2, \phi_3)$ give the magnitude of the respective basis generators $\bm{\sigma} = (X, Y, Z)$. The original parameterization in $\btheta$ with norm $\vert\theta_1\vert + \vert\theta_2\vert + \vert\theta_3\vert = 1$ gives a surface of possible values of $\bphi$ and therefore possible unitary evolutions. The $\mathbb{SU}(2)$ gate (blue) is unbiased because its parameterization gives the correspondence $\btheta = \bphi$ with normalization $\bphi_1^2 + \bphi_2^2 + \bphi_3^2 = 1$. The unitary evolution for the ZYZ decomposition (red) is biased because the surface in the $\bphi$ coordinates does not maintain an equal magnitude in all directions.}
    \label{fig:bias_comparison_ZYZ_SU2}
\end{figure}

The effect of this bias is demonstrated in \Cref{fig:geodesic_learning} for the simplest case of a single-qubit system with an $\mathbb{SU}(2)$ gate. The optimal parameters of the circuit are those that produce the state that gives the minimum of the cost function $C(\btheta) = - \langle Y \rangle$ (green star). We consider various initial parameters acting on the reference state $\rho=|0\rangle\langle 0|$. The corresponding training paths are shown for each initial parameter vector. The training paths for the decomposed ZYZ circuit are depicted in \Cref{fig:geodesic_learning}(a). As the initial parameter $\btheta_0$ acting on the reference state $\rho$ (purple dots) moves closer to an unstable equilibrium point (orange diamond) the training path becomes increasingly suboptimal. At the unstable equilibrium the only gradient information is directly away from the instability rather than providing information about the direction towards the global minimum. This behavior is further illustrated by the gradient vector field on the Bloch sphere in \Cref{fig:geodesic_learning}(c). For the $\mathbb{SU}(N)$ gate, we see in \Cref{fig:geodesic_learning}(b) that the optimization trajectories follow a direct path to the minimum.

\begin{figure*}
    \centering
    \includegraphics[scale=0.247]{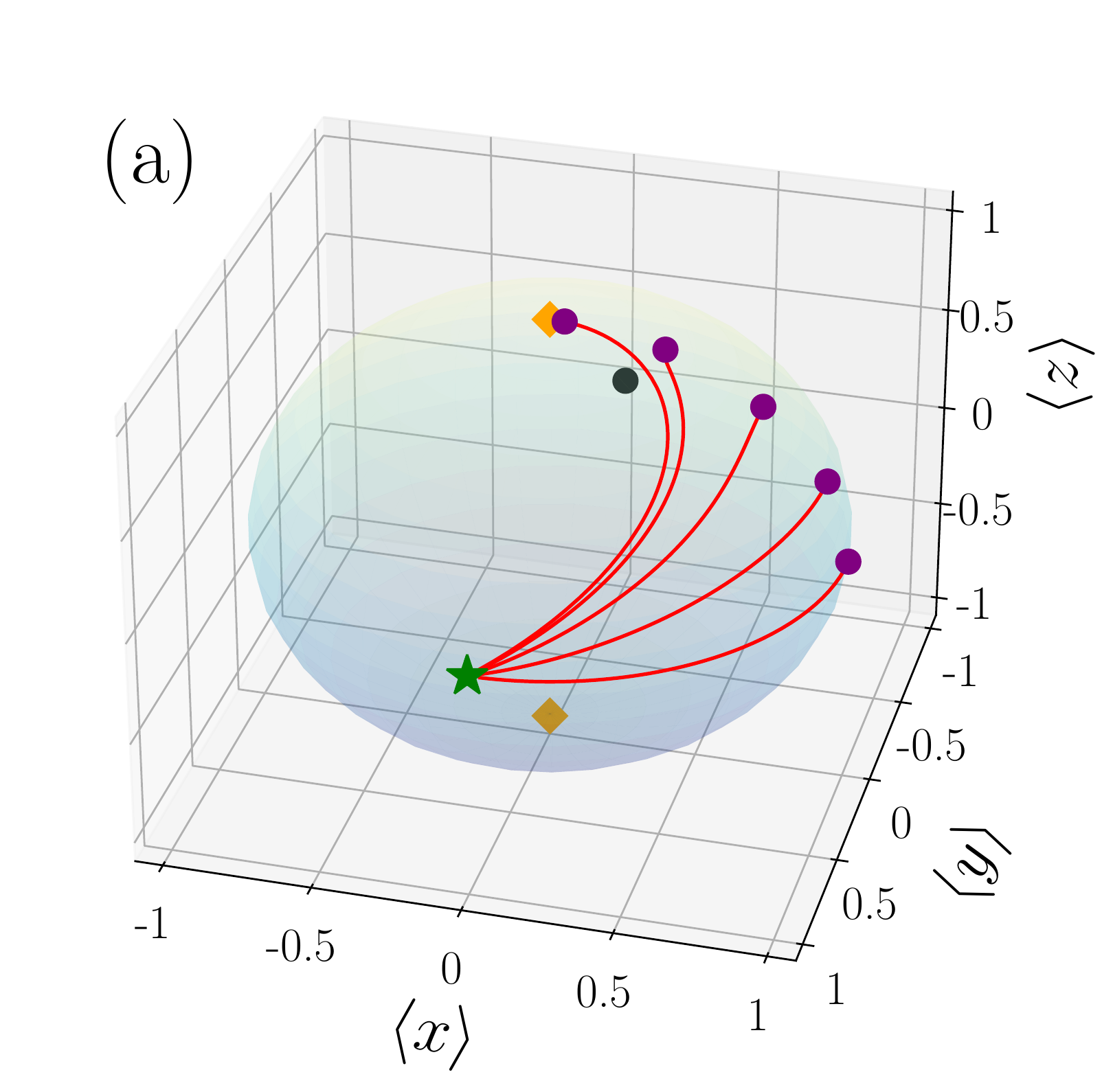}   \includegraphics[scale=0.247]{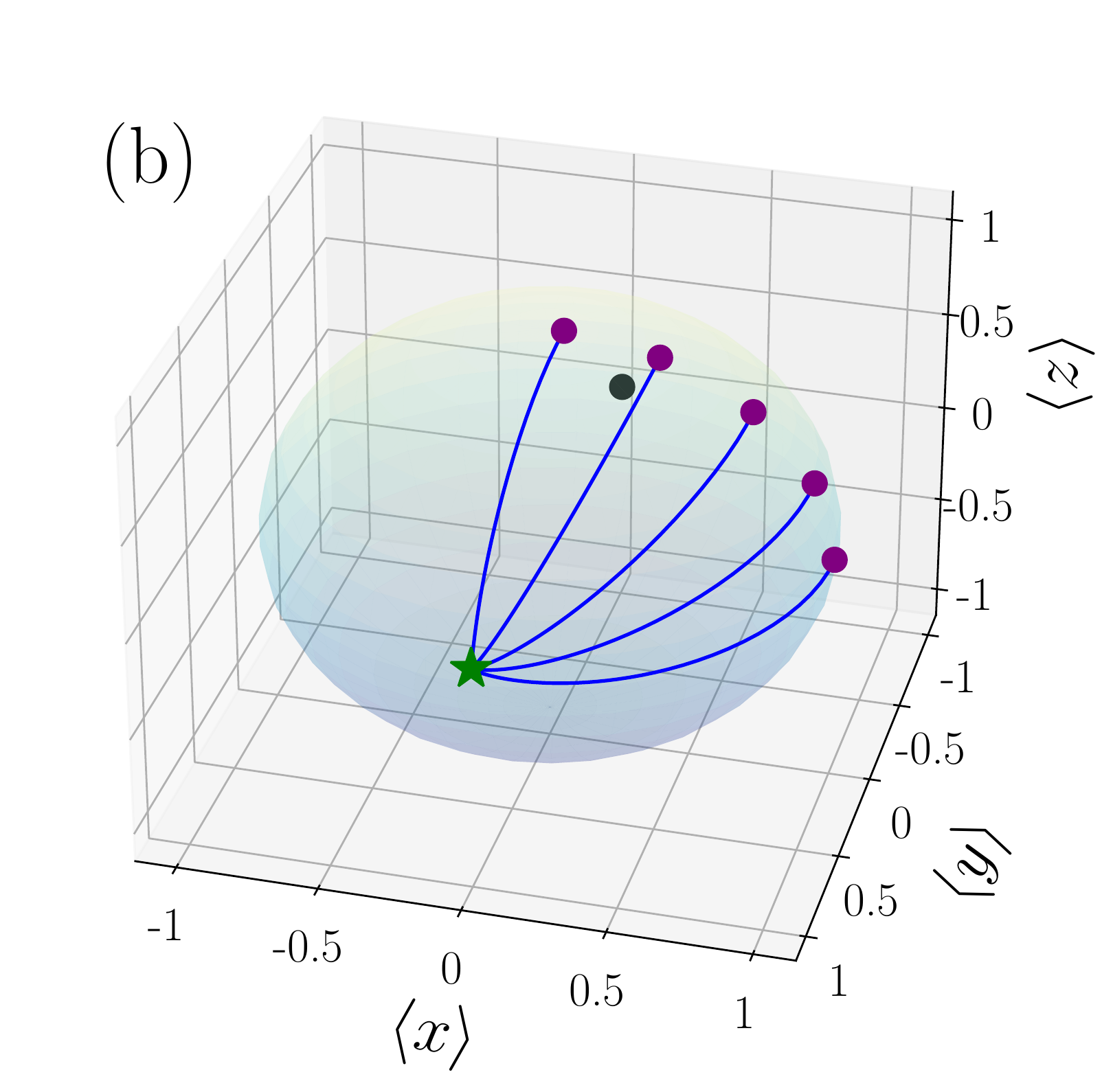}
    \includegraphics[scale=0.247]{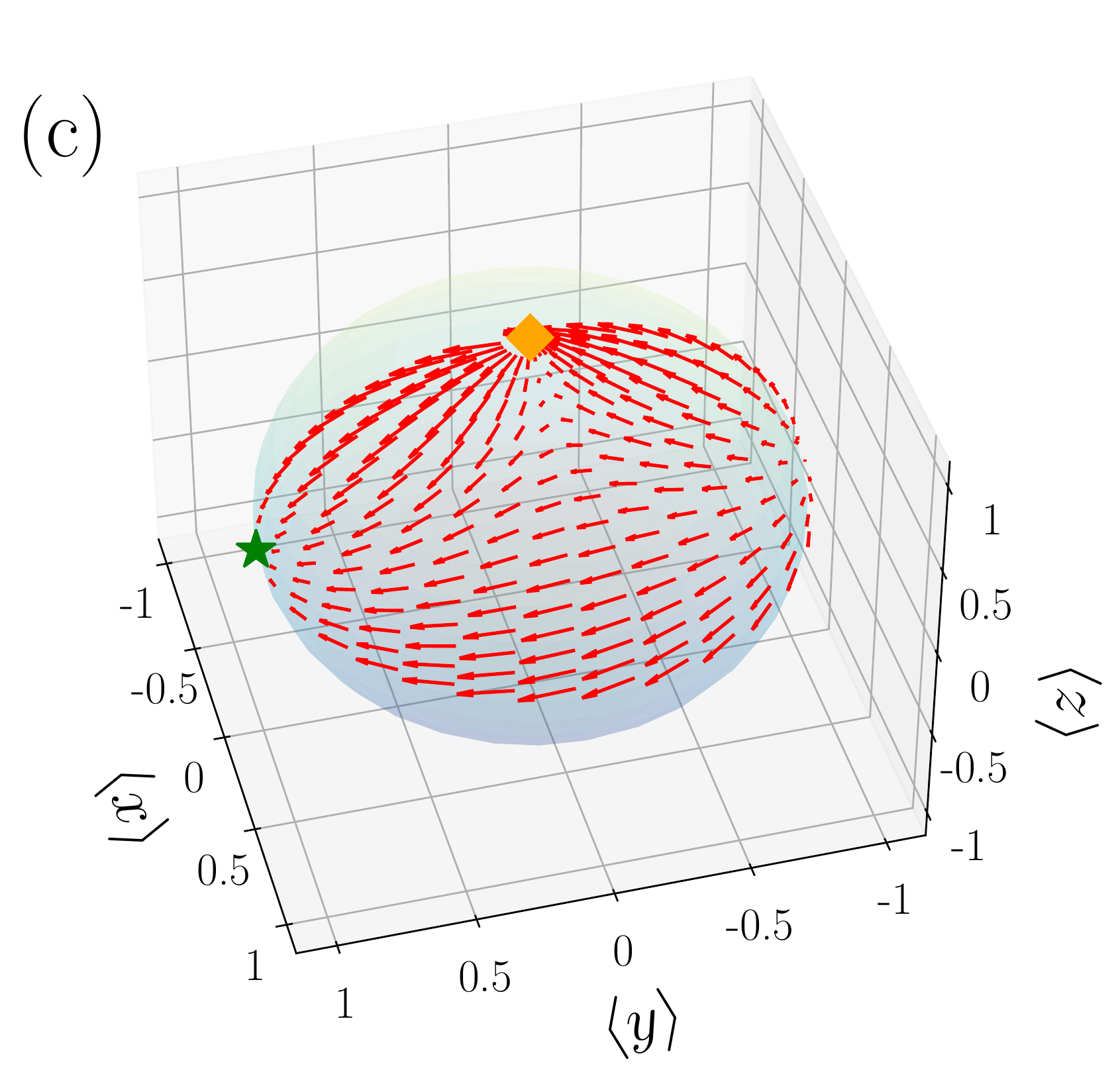}
    \includegraphics[scale=0.247]{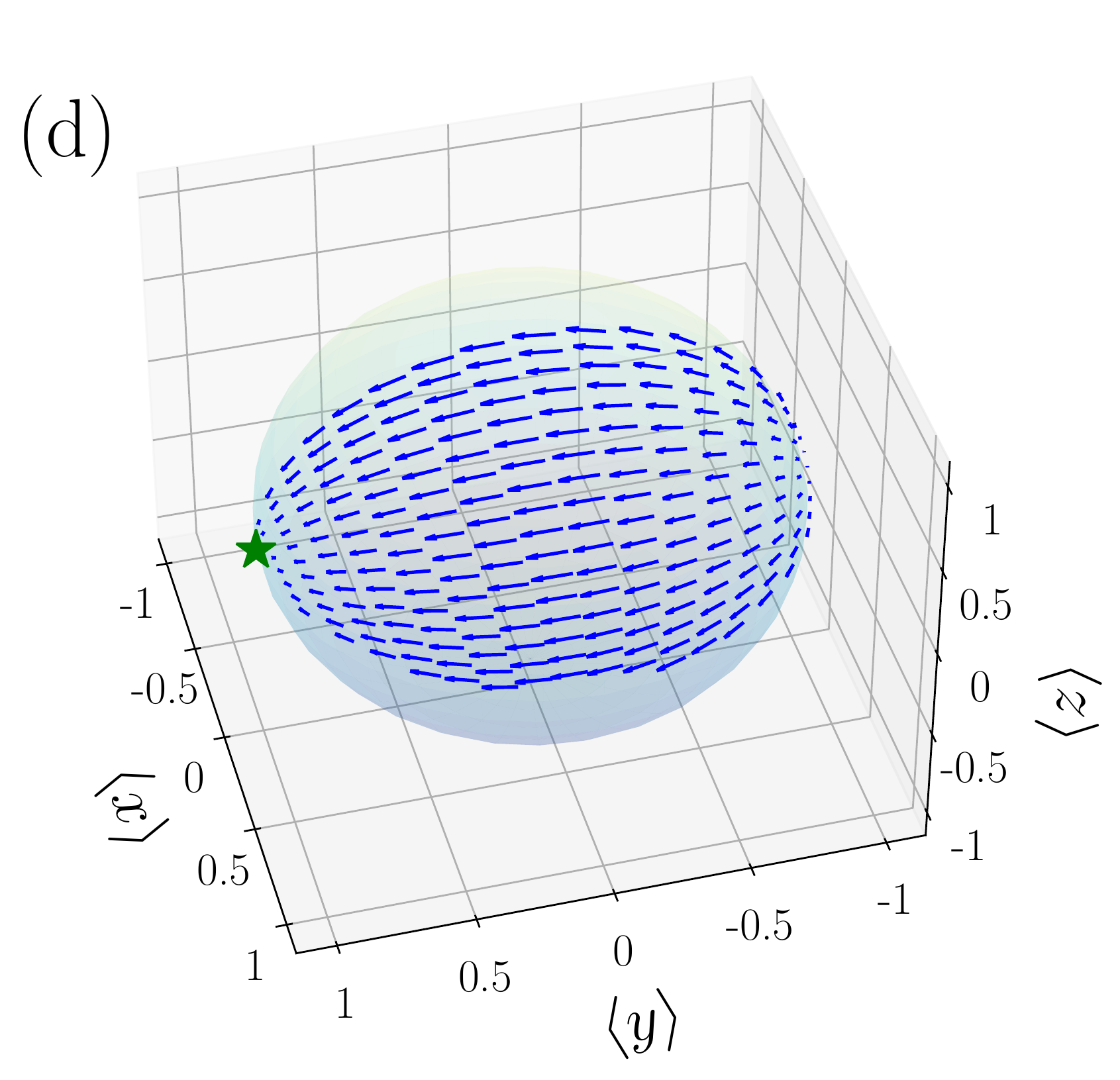}
    \caption{Comparison of the update of circuit parameters from various initial parameters acting on the initial state $\rho = |0\rangle\langle0|$. The training paths are depicted on the Bloch sphere for: (a) parameterized single-qubit rotations for the ZYZ ansatz; and (b) using the $\mathbb{SU}(N)$ gate. The purple dots represent initial states generated by applying $U(\btheta_0)$ with $\btheta_0 = (0, a, 0)$ where $a \in  \left\{ \frac{\pi}{64}, \frac{\pi}{8}, \frac{2 \pi}{8}, \frac{3 \pi}{8}, \frac{\pi}{2} \right\}$ to $\rho$. Note that for this choice of initial parameters, $U(\btheta_0) = V(\btheta_0)$. The objective function is $C(\btheta) = - \langle Y \rangle$, giving the target final state at the green star---the state that gives the global minimum of $C(\btheta)$. The unstable equilibrium points are given by orange diamonds, at $(0,0,1)$ and $(0,0,-1)$, and the black point is at the maximum of the cost function, $(0,1,0)$. (c) shows the gradient vector field of the decomposed ZYZ ansatz. The vector field for the $\mathbb{SU}(2)$ gate, shown in (d), coincides with the geodesic flow towards the target final state at all points which satisfies the gate speed limit.}
    \label{fig:geodesic_learning}
\end{figure*}

\subsection{Numerical experiment}

To investigate the effect on the performance of a typical optimization, we study how an $\mathbb{SU}(N)$ gate compares with a decomposed gate in a larger circuit. In \Cref{fig:average} we provide a non-trivial example, where we incorporate our gates into a circuit and show that it performs better than a decomposed $\mathbb{SU}(4)$ gate on a set of random problem instances. We show the individual optimization trajectories in \Cref{fig:trajectories} which illustrate the faster optimization of $\mathbb{SU}(N)$ gates compared to decomposed gates. Like for the examples in \Cref{fig:exact_grad} and \Cref{fig:sampled_grad}, we assume that there is no gate or measurement noise. Additionally, we assume that we can always implement the gate generated by $\Omega_l(\btheta)$, and have control over all Pauli operators $G_m$. In practice, we typically only have access to a fixed set of generators $\mathrm{span}(\{G_m\}) < \mathrm{span}(\su{N})$. If this is the case, then we require a decomposition of $\exp{t \Omega_l(\btheta)}$ in terms of the available operations on the device~\cite{Khaneja2001cartan, Dalessandro2006decompuni}. All numerical results were obtained with \texttt{PennyLane}~\cite{bergholm2018pennylane}, and the $\mathbb{SU}(N)$ gates can be accessed via the \texttt{qml.SpecialUnitary} class.
Although we do not explore this here, one could make use of sparse optimization methods such as stochastic optimization~\cite{sweke2020stochastic, harrow2021low} and frugal optimization~\cite{arrasmith2020operator} for the GPSR subroutine in our algorithm.

\section{Resource estimation}\label{sec:resource}

To obtain the partial derivative in \Cref{eq:circuit_with_omega} in practice we need to estimate the gradient of a circuit that contains a gate generated by $\Omega_l(\btheta)$. As noted in recent works on GPSR rules~\cite{Wierichs2022grads, Izmaylov2021grads, Oleksandr2021grads}, the computational cost of estimating this gradient is related to the spectral gaps of $\Omega_l(\btheta)$. In particular, if $\{\lambda_j \}$ is the set of (possibly degenerate) eigenvalues of $\Omega_l(\btheta)$, we define the set of unique spectral gaps as $\Gamma =\{\abs{\lambda_j - \lambda_{j'}}\}$ where $j'>j$. Note that for $d$ distinct eigenvalues, the number of unique spectral gaps $R$ is at most $R\leq d(d-1)/2$. The total number of parameter-shifted circuits is then $2R$ for a single partial derivative $\partial_{\theta_l} C(\btheta)$

Depending on the generator $\Omega_l(\btheta)$, this complexity can be improved. For instance, in~\cite{Izmaylov2021grads}, a Cartan decomposition is used to improve the number of circuits required from polynomial to linear or even logarithmic in $N$. Additionally, in~\cite{Wierichs2022grads}, the different costs for first- and second-order gradients are determined for specific varational quantum algorithms like QAOA~\cite{farhi2014quantum} and RotoSolve~\cite{vidal2018rotosolve, parrish2019rotosolve, nakanishi2020rotosolve, Ostaszewski2021rotosolve}. Finally, in~\cite{Oleksandr2021grads}, the computational cost of a variety of different gates is investigated in detail and the variance across the parameter regime is studied.

\begin{figure}[htb!]
    \centering
    \includegraphics[width=\columnwidth]{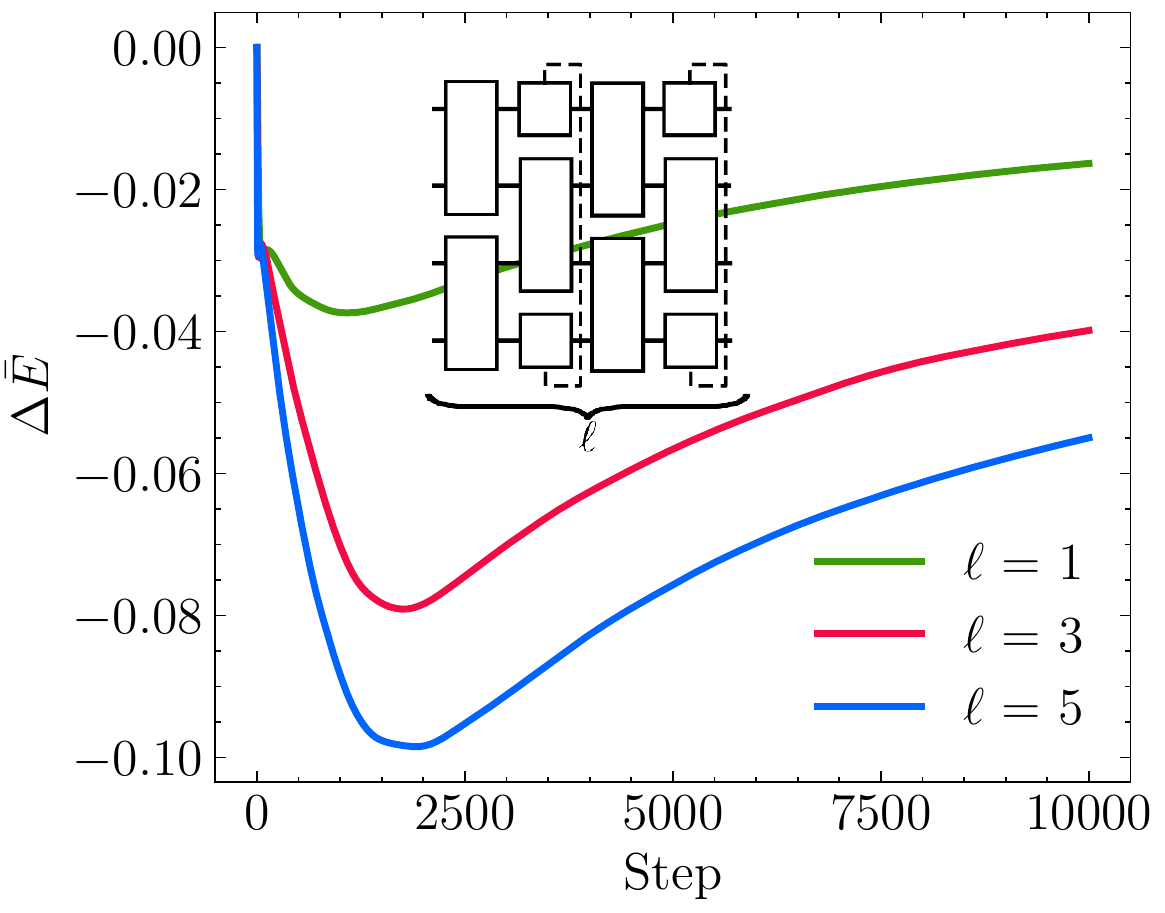}
    \caption{Comparison of decomposed gates versus $\mathbb{SU}(N)$ gates in brick-layer circuits for random 10-qubit Hamiltonians and various depths. We consider the brick-layer circuit indicated  with $\ell=2$ in the inset, with general two-qubit gates acting on the even and odd qubits in each layer. The decomposed gate is the $\SU{4}$ parameterization of~\cite{Vatan2004opt2qubit}, which is optimal in the number of CNOTs required. For each instance, we sample a Hamiltonian from the Gaussian unitary ensemble
    and minimize the cost in \Cref{eq:cost} via standard gradient descent.
    We show the difference of the relative errors in energy $\Bar{E} = (E - E_{\min}) / (E_{\max} - E_{\min})$ between the decomposed gates and the $\mathbb{SU}(N)$ gates, that is $\Delta \bar{E} = \Bar{E}_{\SU{N}} - \Bar{E}_{\mathrm{Decomp.}}$
    The plotted lines are the mean $\Bar{E}$, averaged over 50 random Hamiltonians for each circuit depth $\ell$. We see that for all depths  $\Delta\Bar{E} < 0$ at all points during the optimization, hence the brick-layer circuit with the $\mathbb{SU}(N)$ gates outperforms the circuit where the two-qubit interactions are parametrized as a gate composition.}
    \label{fig:average}
\end{figure}

\begin{figure}[htb!]
    \centering
    \includegraphics[width=\columnwidth]{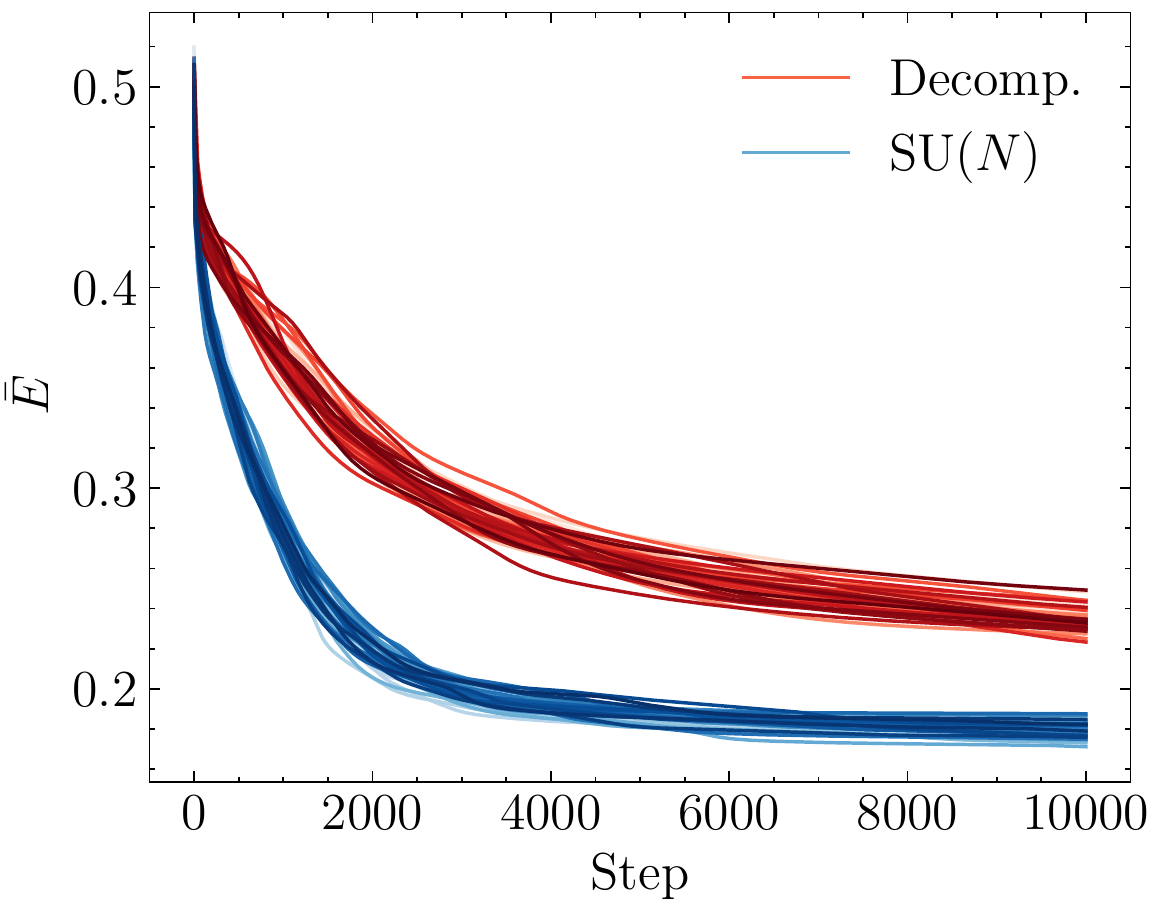}
    \caption{Trajectories from the optimizations in \Cref{fig:average} for 50 random 10-qubit Hamiltonians sampled from the Gaussian unitary ensemble and an $\ell=5$ brick-layer circuit of 2-qubit building blocks. We compare the relative error energy (see \Cref{fig:average} for the definition of $\Bar{E}$) when using a standard gate composition to that when using $\mathbb{SU}(4)$ gates as building blocks.
    The optimization is performed with vanilla gradient descent using a learning rate of $\eta=10^{-3}$.
    The $\mathbb{SU}(4)$ gate consistently leads to faster optimization and better approximations of the ground state energy throughout all $10^5$ optimization steps.}
    \label{fig:trajectories}
\end{figure}

Instead of focusing on specific instances of the generator $\Omega_l(\btheta)$, we make a more general observation about the computational complexity of parameter-shift gradient rules. In general, $\Omega_l(\btheta)$ has full support on $\su{N}$, since the consecutive applications of $\ad{A(\btheta)}$ in \Cref{eq:adj} typically generate all of $\su{N}$~\cite{lloyd1996universal}. However, for specific choices of $A(\btheta)$, the application of $\ad{A(\btheta)^p}$ to $\partial_{ \theta_l} A(\btheta)$ closes to form a subalgebra, called the dynamical Lie algebra of $A(\btheta)$, that is contained in $\su{N}$. These algebras are well-known in the context of quantum optimal control~\cite{Albertini2001dynlie, dalessandro2021}, and have recently been studied in the context of variational quantum algorithms~\cite{larocca2022diagnosing, larocca2021theory}. We define the dynamical Lie algebra (DLA) $\mathcal{L}(A(\btheta))$ as the subalgebra formed under the closure of the non-zero terms in $A(\btheta)$ under the commutator. Ignoring global phases, this will always result in a subalgebra $\mathcal{L}(A(\btheta))\subseteq \su{N}$.
For example, given $A(\btheta) = i(aX + bY)$, $\forall a,b \in\mathbb{R}$, we have $\mathcal{L}(A(\btheta)) = \mathrm{span}\{iX,iY,iZ\}$, since $\ad{X}(Y)=\comm{X}{Y} = iZ$ and successive commutators generate no new contributions. Note that for this example the DLA equals the full Lie algebra $\su{2}$. Explicit constructions of DLAs that span $\mathfrak{so}(N)$ and $\mathfrak{sp}(N)$ are given in~\cite{Schirmer2002dynlie}. In a more recent work, the DLAs of several typical quantum many-body Hamiltonians are studied and their properties are used to prepare efficient time-evolution circuits~\cite{Kokcu2021cartan}. 
In one dimension, the DLAs that are generated by Pauli strings have recently been classified~\cite{wiersema2023classification}.

Interestingly, if the DLA is maximal, i.e., there exists no smaller non-trivial subalgebra within $\mathcal{L}(A(\btheta))$, then the \emph{roots} of the Lie algebra can be related directly to the computational cost of estimating the gradients in \Cref{eq:circuit_with_omega}. We formally establish this connection with the following theorem:

\begin{theorem}\label{th:spectral}
The number of unique spectral gaps $R$ of $\Omega_l(\btheta)$ is upper bounded by the number of roots $|\Phi|$ of any maximal semi-simple DLA,
\begin{align}
    R \leq |\Phi| /2.
\end{align}
\end{theorem}
We provide the proof of \Cref{th:spectral} in App. \ref{app:th4}. We make use of the fact that any semisimple Lie algebra can be written as a direct sum of its weight spaces, which can be identified with its root system~\cite{Serre2000complex}. The number of roots $|\Phi|$ can then be used to bound the total number of unique spectral gaps of $\Omega_l(\btheta)$. We can thus use \Cref{th:spectral} to assess the run time of \Cref{alg:grad}. We give several examples of $\mathbb{SU}(N)$ gates in App. \ref{app:th4} together with the corresponding values of $R$. 
Depending on the physical system or hardware that we are working with, we have to choose a representation for $\su{N}$, which is a map $\su{N}\to\mathfrak{gl}(N, \mathbb{C})$. In \Cref{eq:paulis} we chose this representation to be the tensor product of the fundamental representation, i.e., Pauli monomials. Note however, that \Cref{eq:circuit_with_omega} and \Cref{th:spectral} hold for any irreducible representation of $\su{N}$. 

Additionally, by connecting the spectral gaps to the root system of the DLA, we can make use of a beautiful result in representation theory: the classification of all maximal subalgebras of the classical Lie algebras~\cite{dynkin1957american}. Each root system can be uniquely identified with a particular subalgebra of a Lie algebra and it can be shown that there exist a finite number of root systems. Since a DLA is a subalgebra of $\su{N}$, we can identify all possible DLAs and by extension all possible families of $\mathbb{SU}(N)$ gates. We provide examples of this procedure in App. \ref{app:th4}.

\section{Conclusion}

We have proposed an alternative parameterization of general $\mathbb{SU}(N)$ gates and a method of optimizing these gates in prototypical variational quantum algorithms. We have shown that our gates are more powerful in toy example settings, and motivated why we believe this is the case based on quantum speed-limit arguments. A natural extension of our work would be to test our method in experimental settings, both on gate-based quantum computers or quantum simulators~\cite{Georgescu2014qsim, ebadi2021quantum, Scholl2022xxz}. With regards to the latter, several methods have been investigated that could provide pulse-level optimization of energy cost functions~\cite{ibrahim2022pulse, meitei2020gate}. This would obviate the need for a gate-based model of quantum computing to prepare specific states on quantum hardware. Instead, we work on the Hamiltonian level and control the system directly. Our algorithm could be applied to this setting as well, since we're effectively learning the parameters of some fixed Hamiltonian.

We have shown that the $\mathbb{SU}(N)$ gate in a circuit outperforms a decomposed gate. The number of parameters in our proposed gate equals $4^{N_{\mathrm{qubits}}}$, hence $\mathbb{SU}(N)$ gates acting on a large number of qubits will be impractical. Additionally, it is not clear for which problems one would rather have a deeper circuit with simple gates as opposed to a shallow circuit with more powerful gates. This also begs another question: will our gates suffer from barren plateaus~\cite{mcclean2018barren}? It is likely that a circuit of 2-qubit $\mathbb{SU}(N)$ gates that has linear depth in $N$ will lead to a circuit that forms an approximate 2-design, which will suffer from vanishing gradients. However, appropriate choices of the generators $A(\btheta)$ of our gate could keep the circuit in a polynomially scaling DLA of the entire circuit, which can avoid barren plateaus~\cite{larocca2022diagnosing, larocca2021theory}. Additionally, we can consider parameter initialization strategies that can improve the optimization~\cite{grant2019initialization, skolik2021layerwise}.

Finally, we believe that the connections between variational quantum circuits and representation theory merit further investigation. We connected the classification of all $\mathbb{SU}(N)$ gates with the classification of semisimple Lie algebras. However, this could possibly be extended to a classification of all variational quantum circuits based on the DLA of an ansatz. It seems that the tools to provide such a classification are available and could provide one with a method to assess the trainability and expressivity of variational circuits without explicitly referring to specific ans\"atze.
\section{Acknowledgements}
We want to thank Los Alamos National Lab for their hospitality during the Quantum Computing Summer School where the initial stages of this project took place. RW wants to thank Lex Kemper and Efekan K{\"o}kc{\"u} for discussions on the subject of dynamical Lie algebras and Matt Duchenes for his suggestions with regards to the experimental implications of our work. DL acknowledges support from the EPSRC Centre for Doctoral Training in Delivering Quantum Technologies, grant ref. EP/S021582/1. JFCA acknowledges support from the Natural Sciences and Engineering Research Council
(NSERC), the Shared Hierarchical Academic Research Computing Network (SHARCNET), Compute Canada, and the Canadian Institute for Advanced Research (CIFAR) AI chair program. Resources used in preparing this research were provided, in part, by the Province of Ontario, the Government of Canada through CIFAR, and companies sponsoring the Vector Institute \url{https://vectorinstitute.ai/#partners}.

\bibliographystyle{quantumbib.bst}
\bibliography{library.bib}

\begin{thebibliography}{10}

\bibitem{Cerezo2021vqa}
M.~Cerezo, Andrew Arrasmith, Ryan Babbush, Simon~C. Benjamin, Suguru Endo,
  Keisuke Fujii, Jarrod~R. McClean, Kosuke Mitarai, Xiao Yuan, Lukasz Cincio,
  and Patrick~J. Coles.
\newblock ``{Variational quantum algorithms}''.
\newblock \href{https://dx.doi.org/10.1038/s42254-021-00348-9}{Nature Reviews
  Physics {\bf 3}, 625--644}~(2021).

\bibitem{Tilly2022vqe}
Jules Tilly, Hongxiang Chen, Shuxiang Cao, Dario Picozzi, Kanav Setia, Ying Li,
  Edward Grant, Leonard Wossnig, Ivan Rungger, George~H. Booth, and Jonathan
  Tennyson.
\newblock ``{The Variational Quantum Eigensolver: A review of methods and best
  practices}''.
\newblock
  \href{https://dx.doi.org/https://doi.org/10.1016/j.physrep.2022.08.003}{Physics
  Reports {\bf 986}, 1--128}~(2022).

\bibitem{Li2017grads}
Jun Li, Xiaodong Yang, Xinhua Peng, and Chang-Pu Sun.
\newblock ``{Hybrid Quantum-Classical Approach to Quantum Optimal Control}''.
\newblock \href{https://dx.doi.org/10.1103/PhysRevLett.118.150503}{Phys. Rev.
  Lett. {\bf 118}, 150503}~(2017).

\bibitem{Mitarai2018grads}
K.~Mitarai, M.~Negoro, M.~Kitagawa, and K.~Fujii.
\newblock ``{Quantum circuit learning}''.
\newblock \href{https://dx.doi.org/10.1103/PhysRevA.98.032309}{Phys. Rev. A
  {\bf 98}, 032309}~(2018).

\bibitem{Schuld2019grads}
Maria Schuld, Ville Bergholm, Christian Gogolin, Josh Izaac, and Nathan
  Killoran.
\newblock ``{Evaluating analytic gradients on quantum hardware}''.
\newblock \href{https://dx.doi.org/10.1103/PhysRevA.99.032331}{Phys. Rev. A
  {\bf 99}, 032331}~(2019).

\bibitem{Crooks2019grads}
Gavin~E. Crooks.
\newblock ``{Gradients of parameterized quantum gates using the parameter-shift
  rule and gate decomposition}''~(2019)
  \href{http://arxiv.org/abs/1905.13311}{arXiv:1905.13311}.

\bibitem{Izmaylov2021grads}
Artur~F. Izmaylov, Robert~A. Lang, and Tzu-Ching Yen.
\newblock ``{Analytic gradients in variational quantum algorithms: Algebraic
  extensions of the parameter-shift rule to general unitary transformations}''.
\newblock \href{https://dx.doi.org/10.1103/PhysRevA.104.062443}{Phys. Rev. A
  {\bf 104}, 062443}~(2021).

\bibitem{Wierichs2022grads}
David Wierichs, Josh Izaac, Cody Wang, and Cedric Yen-Yu Lin.
\newblock ``{General parameter-shift rules for quantum gradients}''.
\newblock \href{https://dx.doi.org/10.22331/q-2022-03-30-677}{{Quantum} {\bf
  6}, 677}~(2022).

\bibitem{Oleksandr2021grads}
Oleksandr Kyriienko and Vincent~E. Elfving.
\newblock ``{Generalized quantum circuit differentiation rules}''.
\newblock \href{https://dx.doi.org/10.1103/PhysRevA.104.052417}{Phys. Rev. A
  {\bf 104}, 052417}~(2021).

\bibitem{Theis2023propershiftrules}
Dirk~Oliver Theis.
\newblock ``"{P}roper" {S}hift {R}ules for {D}erivatives of
  {P}erturbed-{P}arametric {Q}uantum {E}volutions''.
\newblock \href{https://dx.doi.org/10.22331/q-2023-07-11-1052}{{Quantum} {\bf
  7}, 1052}~(2023).

\bibitem{Slattery2022uniblock}
Lucas Slattery, Benjamin Villalonga, and Bryan~K. Clark.
\newblock ``{Unitary block optimization for variational quantum algorithms}''.
\newblock \href{https://dx.doi.org/10.1103/PhysRevResearch.4.023072}{Phys. Rev.
  Research {\bf 4}, 023072}~(2022).

\bibitem{Liu2019fewer}
Jin-Guo Liu, Yi-Hong Zhang, Yuan Wan, and Lei Wang.
\newblock ``{Variational quantum eigensolver with fewer qubits}''.
\newblock \href{https://dx.doi.org/10.1103/PhysRevResearch.1.023025}{Phys. Rev.
  Research {\bf 1}, 023025}~(2019).

\bibitem{Abhinav2017hardware}
Abhinav Kandala, Antonio Mezzacapo, Kristan Temme, Maika Takita, Markus Brink,
  Jerry~M. Chow, and Jay~M. Gambetta.
\newblock ``{Hardware-efficient variational quantum eigensolver for small
  molecules and quantum magnets}''.
\newblock \href{https://dx.doi.org/10.1038/nature23879}{Nature {\bf 549},
  242--246}~(2017).

\bibitem{Khaneja2001cartan}
Navin Khaneja and Steffen~J. Glaser.
\newblock ``{Cartan decomposition of $SU(2^n)$and control of spin systems}''.
\newblock
  \href{https://dx.doi.org/https://doi.org/10.1016/S0301-0104(01)00318-4}{Chemical
  Physics {\bf 267}, 11--23}~(2001).

\bibitem{kraus2001optimal}
Barbara Kraus and Juan~I Cirac.
\newblock ``{Optimal creation of entanglement using a two-qubit gate}''.
\newblock \href{https://dx.doi.org/10.1103/PhysRevA.63.062309}{Physical Review
  A {\bf 63}, 062309}~(2001).

\bibitem{Vatan2004opt2qubit}
Farrokh Vatan and Colin Williams.
\newblock ``{Optimal quantum circuits for general two-qubit gates}''.
\newblock \href{https://dx.doi.org/10.1103/PhysRevA.69.032315}{Phys. Rev. A
  {\bf 69}, 032315}~(2004).

\bibitem{vatan2004realization}
Farrokh Vatan and Colin~P Williams.
\newblock ``{Realization of a general three-qubit quantum gate}''~(2004).
\newblock
  \href{http://arxiv.org/abs/quant-ph/0401178}{arXiv:quant-ph/0401178}.

\bibitem{Vartiainen2004}
Juha~J. Vartiainen, Mikko M\"ott\"onen, and Martti~M. Salomaa.
\newblock ``{Efficient Decomposition of Quantum Gates}''.
\newblock \href{https://dx.doi.org/10.1103/PhysRevLett.92.177902}{Phys. Rev.
  Lett. {\bf 92}, 177902}~(2004).

\bibitem{Dalessandro2006decompuni}
Domenico D’Alessandro and Raffaele Romano.
\newblock ``{Decompositions of unitary evolutions and entanglement dynamics of
  bipartite quantum systems}''.
\newblock \href{https://dx.doi.org/10.1063/1.2245205}{Journal of Mathematical
  Physics {\bf 47}, 082109}~(2006).

\bibitem{Zulehner2019compiling}
Alwin Zulehner and Robert Wille.
\newblock ``{Compiling SU(4) Quantum Circuits to IBM QX Architectures}''.
\newblock In Proceedings of the 24th Asia and South Pacific Design Automation
  Conference.
\newblock \href{https://dx.doi.org/10.1145/3287624.3287704}{Page 185–190}.
\newblock ASPDAC '19New York, NY, USA~(2019). Association for Computing
  Machinery.

\bibitem{Foxen2020compiling}
B.~Foxen, C.~Neill, A.~Dunsworth, P.~Roushan, B.~Chiaro, A.~Megrant, J.~Kelly,
  Zijun Chen, K.~Satzinger, R.~Barends, F.~Arute, K.~Arya, R.~Babbush,
  D.~Bacon, J.~C. Bardin, S.~Boixo, D.~Buell, B.~Burkett, Yu~Chen, R.~Collins,
  E.~Farhi, A.~Fowler, C.~Gidney, M.~Giustina, R.~Graff, M.~Harrigan, T.~Huang,
  S.~V. Isakov, E.~Jeffrey, Z.~Jiang, D.~Kafri, K.~Kechedzhi, P.~Klimov,
  A.~Korotkov, F.~Kostritsa, D.~Landhuis, E.~Lucero, J.~McClean, M.~McEwen,
  X.~Mi, M.~Mohseni, J.~Y. Mutus, O.~Naaman, M.~Neeley, M.~Niu, A.~Petukhov,
  C.~Quintana, N.~Rubin, D.~Sank, V.~Smelyanskiy, A.~Vainsencher, T.~C. White,
  Z.~Yao, P.~Yeh, A.~Zalcman, H.~Neven, and J.~M. Martinis.
\newblock ``{Demonstrating a Continuous Set of Two-Qubit Gates for Near-Term
  Quantum Algorithms}''.
\newblock \href{https://dx.doi.org/10.1103/PhysRevLett.125.120504}{Phys. Rev.
  Lett. {\bf 125}, 120504}~(2020).

\bibitem{groeneveld1994reparameterization}
E~Groeneveld.
\newblock ``{A reparameterization to improve numerical optimization in
  multivariate REML (co) variance component estimation}''.
\newblock \href{https://dx.doi.org/10.1186/1297-9686-26-6-537}{Genetics
  Selection Evolution {\bf 26}, 537--545}~(1994).

\bibitem{raiko2012deep}
Tapani Raiko, Harri Valpola, and Yann Lecun.
\newblock ``Deep learning made easier by linear transformations in
  perceptrons''.
\newblock In Neil~D. Lawrence and Mark Girolami, editors, Proceedings of the
  Fifteenth International Conference on Artificial Intelligence and Statistics.
\newblock Volume~22 of Proceedings of Machine Learning Research, pages
  924--932.
\newblock La Palma, Canary Islands~(2012). PMLR.
\newblock  url:~\url{https://proceedings.mlr.press/v22/raiko12.html}.

\bibitem{ioffe2015batch}
Sergey Ioffe and Christian Szegedy.
\newblock ``{Batch normalization: Accelerating deep network training by
  reducing internal covariate shift}''.
\newblock In International conference on machine learning.
\newblock \href{https://dx.doi.org/10.5555/3045118.3045167}{Pages 448--456}.
\newblock PMLR~(2015).

\bibitem{salimans2016weight}
Tim Salimans and Durk~P Kingma.
\newblock ``{Weight normalization: A simple reparameterization to accelerate
  training of deep neural networks}''.
\newblock In Advances in neural information processing systems.
\newblock
  \href{https://dx.doi.org/https://doi.org/10.48550/arXiv.1602.07868}{Volume~29}.
\newblock ~(2016).

\bibitem{price1958useful}
Robert Price.
\newblock ``{A useful theorem for nonlinear devices having Gaussian inputs}''.
\newblock \href{https://dx.doi.org/10.1109/TIT.1958.1057444}{IRE Transactions
  on Information Theory {\bf 4}, 69--72}~(1958).

\bibitem{rezende2014stochastic}
Danilo~Jimenez Rezende, Shakir Mohamed, and Daan Wierstra.
\newblock ``Stochastic backpropagation and approximate inference in deep
  generative models''.
\newblock In Eric~P. Xing and Tony Jebara, editors, Proceedings of the 31st
  International Conference on Machine Learning.
\newblock Volume~32 of Proceedings of Machine Learning Research, pages
  1278--1286.
\newblock Bejing, China~(2014). PMLR.
\newblock  url:~\url{https://proceedings.mlr.press/v32/rezende14.html}.

\bibitem{welling2014vae}
Diederik~P. Kingma and Max Welling.
\newblock ``{Auto-Encoding Variational Bayes}''.
\newblock In Yoshua Bengio and Yann LeCun, editors, 2nd International
  Conference on Learning Representations, {ICLR} 2014, Banff, AB, Canada, April
  14-16, 2014, Conference Track Proceedings.
\newblock ~(2014).
\newblock  url:~\url{http://arxiv.org/abs/1312.6114}.

\bibitem{hall2013lie}
Brian~C Hall.
\newblock ``{Lie groups, Lie algebras, and representations}''.
\newblock \href{https://dx.doi.org/10.1007/978-3-319-13467-3}{Springer}.
  ~(2013).
\newblock 2nd edition.

\bibitem{fulton2013representation}
William Fulton and Joe Harris.
\newblock ``{Representation theory: a first course}''.
\newblock \href{https://dx.doi.org/10.1007/978-1-4612-0979-9}{Volume 129}.
\newblock Springer Science \& Business Media. ~(2013).

\bibitem{rossmann2002lie}
W.~Rossmann.
\newblock ``{Lie Groups: An Introduction Through Linear Groups}''.
\newblock \href{https://dx.doi.org/10.1093/oso/9780198596837.001.0001}{Oxford
  graduate texts in mathematics}. Oxford University Press. ~(2002).
\newblock 5th edition.

\bibitem{Serre2009lie}
Jean-Pierre Serre.
\newblock ``{Lie algebras and Lie groups: 1964 lectures given at Harvard
  University}''.
\newblock \href{https://dx.doi.org/10.1007/978-3-540-70634-2}{Springer}.
  ~(2009).

\bibitem{Schuch2003gates}
Norbert Schuch and Jens Siewert.
\newblock ``Natural two-qubit gate for quantum computation using the
  $\mathrm{XY}$ interaction''.
\newblock \href{https://dx.doi.org/10.1103/PhysRevA.67.032301}{Phys. Rev. A
  {\bf 67}, 032301}~(2003).

\bibitem{Orlando1999supercond}
T.~P. Orlando, J.~E. Mooij, Lin Tian, Caspar~H. van~der Wal, L.~S. Levitov,
  Seth Lloyd, and J.~J. Mazo.
\newblock ``Superconducting persistent-current qubit''.
\newblock \href{https://dx.doi.org/10.1103/PhysRevB.60.15398}{Phys. Rev. B {\bf
  60}, 15398--15413}~(1999).

\bibitem{Kane1998nuc}
B.~E. Kane.
\newblock ``A silicon-based nuclear spin quantum computer''.
\newblock \href{https://dx.doi.org/10.1038/30156}{Nature {\bf 393},
  133--137}~(1998).

\bibitem{Imamoglu1999qcavity}
A.~Imamog\ifmmode\bar\else\textasciimacron\fi{}lu, D.~D. Awschalom, G.~Burkard,
  D.~P. DiVincenzo, D.~Loss, M.~Sherwin, and A.~Small.
\newblock ``Quantum information processing using quantum dot spins and cavity
  qed''.
\newblock \href{https://dx.doi.org/10.1103/PhysRevLett.83.4204}{Phys. Rev.
  Lett. {\bf 83}, 4204--4207}~(1999).

\bibitem{leng2022differentiable}
Jiaqi Leng, Yuxiang Peng, Yi-Ling Qiao, Ming Lin, and Xiaodi Wu.
\newblock ``{Differentiable Analog Quantum Computing for Optimization and
  Control}''~(2022).
\newblock  \href{http://arxiv.org/abs/2210.15812}{arXiv:2210.15812}.

\bibitem{Wilcox1967exponential}
R.~M. Wilcox.
\newblock ``{Exponential Operators and Parameter Differentiation in Quantum
  Physics}''.
\newblock \href{https://dx.doi.org/10.1063/1.1705306}{Journal of Mathematical
  Physics {\bf 8}, 962--982}~(1967).
\newblock
  \href{http://arxiv.org/abs/https://doi.org/10.1063/1.1705306}{arXiv:https://doi.org/10.1063/1.1705306}.

\bibitem{whittaker1915}
E.~T. Whittaker.
\newblock ``{XVIII.—On the Functions which are represented by the Expansions
  of the Interpolation-Theory}''.
\newblock \href{https://dx.doi.org/10.1017/S0370164600017806}{Proceedings of
  the Royal Society of Edinburgh {\bf 35}, 181–194}~(1915).

\bibitem{Jax2018github}
James Bradbury, Roy Frostig, Peter Hawkins, Matthew~James Johnson, Chris Leary,
  Dougal Maclaurin, George Necula, Adam Paszke, Jake Vander{P}las, Skye
  Wanderman-{M}ilne, and Qiao Zhang~(2018).
\newblock  code:~\href{https://github.com/google/jax}{google/jax}.

\bibitem{paszke2019pytorch}
Adam Paszke, Sam Gross, Francisco Massa, Adam Lerer, James Bradbury, Gregory
  Chanan, Trevor Killeen, Zeming Lin, Natalia Gimelshein, Luca Antiga, et~al.
\newblock ``{Pytorch: An imperative style, high-performance deep learning
  library}''.
\newblock In Advances in neural information processing systems.
\newblock \href{https://dx.doi.org/10.48550/arXiv.1912.01703}{Volume~32}.
\newblock ~(2019).

\bibitem{tensorflow2015}
Mart\'{i}n Abadi, Ashish Agarwal, Paul Barham, Eugene Brevdo, Zhifeng Chen,
  Craig Citro, Greg~S. Corrado, Andy Davis, Jeffrey Dean, Matthieu Devin,
  Sanjay Ghemawat, Ian Goodfellow, Andrew Harp, Geoffrey Irving, Michael Isard,
  Yangqing Jia, Rafal Jozefowicz, Lukasz Kaiser, Manjunath Kudlur, Josh
  Levenberg, Dandelion Man\'{e}, Rajat Monga, Sherry Moore, Derek Murray, Chris
  Olah, Mike Schuster, Jonathon Shlens, Benoit Steiner, Ilya Sutskever, Kunal
  Talwar, Paul Tucker, Vincent Vanhoucke, Vijay Vasudevan, Fernanda Vi\'{e}gas,
  Oriol Vinyals, Pete Warden, Martin Wattenberg, Martin Wicke, Yuan Yu, and
  Xiaoqiang Zheng~(2015).
\newblock
  code:~\href{https://www.tensorflow.org/}{https://www.tensorflow.org/}.

\bibitem{expm}
{A JAX implementation of the matrix exponential that can be differentiated via
  automatic differentiation:
  \url{https://jax.readthedocs.io/en/latest/_autosummary/jax.scipy.linalg.expm.html}}.

\bibitem{al2010new}
Awad~H Al-Mohy and Nicholas~J Higham.
\newblock ``{A new scaling and squaring algorithm for the matrix
  exponential}''.
\newblock \href{https://dx.doi.org/10.1137/09074721}{SIAM Journal on Matrix
  Analysis and Applications {\bf 31}, 970--989}~(2010).

\bibitem{Banchi2021measuringanalytic}
Leonardo Banchi and Gavin~E. Crooks.
\newblock ``{Measuring {A}nalytic {G}radients of {G}eneral {Q}uantum
  {E}volution with the {S}tochastic {P}arameter {S}hift {R}ule}''.
\newblock \href{https://dx.doi.org/10.22331/q-2021-01-25-386}{{Quantum} {\bf
  5}, 386}~(2021).

\bibitem{bittel2022fast}
Lennart Bittel, Jens Watty, and Martin Kliesch.
\newblock ``{Fast gradient estimation for variational quantum
  algorithms}''~(2022).
\newblock  \href{http://arxiv.org/abs/2210.06484}{arXiv:2210.06484}.

\bibitem{our_code}
Roeland Wiersema, Dylan Lewis, David Wierichs, Juan Carrasquilla, and Nathan
  Killoran~(2023).
\newblock
  code:~\href{https://github.com/dwierichs/Here-comes-the-SUN}{dwierichs/Here-comes-the-SUN}.

\bibitem{SchulteHerbruggen2010gradflow}
Thomas Schulte-Herbr\"uggen, Steffen~j. Glaser, Gunther Dirr, and Uwe Helmke.
\newblock ``{Gradient Flows for Optimization in Quantum Information and Quantum
  Dynamics: Foundations and Applications}''.
\newblock \href{https://dx.doi.org/10.1142/S0129055X10004053}{Reviews in
  Mathematical Physics {\bf 22}, 597--667}~(2010).

\bibitem{Wiersema2022riemann}
Roeland Wiersema and Nathan Killoran.
\newblock ``Optimizing quantum circuits with riemannian gradient flow''~(2023).

\bibitem{bergholm2018pennylane}
Ville Bergholm, Josh Izaac, Maria Schuld, Christian Gogolin, M~Sohaib Alam,
  Shahnawaz Ahmed, Juan~Miguel Arrazola, Carsten Blank, Alain Delgado, Soran
  Jahangiri, et~al.
\newblock ``{Pennylane: Automatic differentiation of hybrid quantum-classical
  computations}''~(2018).
\newblock  \href{http://arxiv.org/abs/1811.04968}{arXiv:1811.04968}.

\bibitem{sweke2020stochastic}
Ryan Sweke, Frederik Wilde, Johannes Meyer, Maria Schuld, Paul~K. Faehrmann,
  Barth{\'{e}}l{\'{e}}my Meynard-Piganeau, and Jens Eisert.
\newblock ``{Stochastic gradient descent for hybrid quantum-classical
  optimization}''.
\newblock \href{https://dx.doi.org/10.22331/q-2020-08-31-314}{{Quantum} {\bf
  4}, 314}~(2020).

\bibitem{harrow2021low}
Aram~W. Harrow and John~C. Napp.
\newblock ``{Low-Depth Gradient Measurements Can Improve Convergence in
  Variational Hybrid Quantum-Classical Algorithms}''.
\newblock \href{https://dx.doi.org/10.1103/PhysRevLett.126.140502}{Phys. Rev.
  Lett. {\bf 126}, 140502}~(2021).

\bibitem{arrasmith2020operator}
Andrew Arrasmith, Lukasz Cincio, Rolando~D Somma, and Patrick~J Coles.
\newblock ``{Operator sampling for shot-frugal optimization in variational
  algorithms}''~(2020).
\newblock  \href{http://arxiv.org/abs/2004.06252}{arXiv:2004.06252}.

\bibitem{farhi2014quantum}
Edward Farhi, Jeffrey Goldstone, and Sam Gutmann.
\newblock ``{A quantum approximate optimization algorithm}''~(2014).
\newblock  \href{http://arxiv.org/abs/1411.4028}{arXiv:1411.4028}.

\bibitem{vidal2018rotosolve}
Javier~Gil Vidal and Dirk~Oliver Theis.
\newblock ``{Calculus on parameterized quantum circuits}''~(2018).
\newblock  \href{http://arxiv.org/abs/1812.06323}{arXiv:1812.06323}.

\bibitem{parrish2019rotosolve}
Robert~M Parrish, Joseph~T Iosue, Asier Ozaeta, and Peter~L McMahon.
\newblock ``{A Jacobi diagonalization and Anderson acceleration algorithm for
  variational quantum algorithm parameter optimization}''~(2019).
\newblock  \href{http://arxiv.org/abs/1904.03206}{arXiv:1904.03206}.

\bibitem{nakanishi2020rotosolve}
Ken~M. Nakanishi, Keisuke Fujii, and Synge Todo.
\newblock ``{Sequential minimal optimization for quantum-classical hybrid
  algorithms}''.
\newblock \href{https://dx.doi.org/10.1103/PhysRevResearch.2.043158}{Phys. Rev.
  Res. {\bf 2}, 043158}~(2020).

\bibitem{Ostaszewski2021rotosolve}
Mateusz Ostaszewski, Edward Grant, and Marcello Benedetti.
\newblock ``{Structure optimization for parameterized quantum circuits}''.
\newblock \href{https://dx.doi.org/10.22331/q-2021-01-28-391}{{Quantum} {\bf
  5}, 391}~(2021).

\bibitem{lloyd1996universal}
Seth Lloyd.
\newblock ``{Universal quantum simulators}''.
\newblock \href{https://dx.doi.org/10.1126/science.273.5278.1073}{Science {\bf
  273}, 1073--1078}~(1996).

\bibitem{Albertini2001dynlie}
F.~Albertini and D.~D'Alessandro.
\newblock ``{Notions of controllability for quantum mechanical systems}''.
\newblock In Proceedings of the 40th IEEE Conference on Decision and Control
  (Cat. No.01CH37228).
\newblock \href{https://dx.doi.org/10.1109/CDC.2001.981126}{Volume~2, pages
  1589--1594 vol.2}.
\newblock ~(2001).

\bibitem{dalessandro2021}
Domenico d'Alessandro.
\newblock ``{Introduction to quantum control and dynamics}''.
\newblock \href{https://dx.doi.org/10.1201/9781003051268}{Chapman and
  hall/CRC}. ~(2021).
\newblock 2nd edition.

\bibitem{larocca2022diagnosing}
Martin Larocca, Piotr Czarnik, Kunal Sharma, Gopikrishnan Muraleedharan,
  Patrick~J. Coles, and M.~Cerezo.
\newblock ``{Diagnosing {B}arren {P}lateaus with {T}ools from {Q}uantum
  {O}ptimal {C}ontrol}''.
\newblock \href{https://dx.doi.org/10.22331/q-2022-09-29-824}{{Quantum} {\bf
  6}, 824}~(2022).

\bibitem{larocca2021theory}
Mart{\'\i}n Larocca, Nathan Ju, Diego Garc{\'\i}a-Mart{\'\i}n, Patrick~J.
  Coles, and Marco Cerezo.
\newblock ``Theory of overparametrization in quantum neural networks''.
\newblock \href{https://dx.doi.org/10.1038/s43588-023-00467-6}{Nature
  Computational Science {\bf 3}, 542--551}~(2023).

\bibitem{Schirmer2002dynlie}
S~G Schirmer, I~C~H Pullen, and A~I Solomon.
\newblock ``{Identification of dynamical Lie algebras for finite-level quantum
  control systems}''.
\newblock \href{https://dx.doi.org/10.1088/0305-4470/35/9/319}{Journal of
  Physics A: Mathematical and General {\bf 35}, 2327}~(2002).

\bibitem{Kokcu2021cartan}
Efekan K\"okc\"u, Thomas Steckmann, Yan Wang, J.~K. Freericks, Eugene~F.
  Dumitrescu, and Alexander~F. Kemper.
\newblock ``Fixed depth hamiltonian simulation via cartan decomposition''.
\newblock \href{https://dx.doi.org/10.1103/PhysRevLett.129.070501}{Phys. Rev.
  Lett. {\bf 129}, 070501}~(2022).

\bibitem{wiersema2023classification}
Roeland Wiersema, Efekan K{\"o}kc{\"u}, Alexander~F Kemper, and Bojko~N
  Bakalov.
\newblock ``Classification of dynamical lie algebras for translation-invariant
  2-local spin systems in one dimension''~(2023).
\newblock  \href{http://arxiv.org/abs/2203.05690}{arXiv:2203.05690}.

\bibitem{Serre2000complex}
Jean-Pierre Serre.
\newblock ``{Complex semisimple Lie algebras}''.
\newblock \href{https://dx.doi.org/10.1007/978-3-642-56884-8}{Springer Science
  \& Business Media}. ~(2000).
\newblock 1st edition.

\bibitem{dynkin1957american}
Eugene~Borisovich Dynkin.
\newblock ``{American Mathematical Society Translations: Five Papers on Algebra
  and Group Theory}''.
\newblock \href{https://dx.doi.org/10.1090/trans2/006}{American Mathematical
  Society}. ~(1957).

\bibitem{Georgescu2014qsim}
I.~M. Georgescu, S.~Ashhab, and Franco Nori.
\newblock ``{Quantum simulation}''.
\newblock \href{https://dx.doi.org/10.1103/RevModPhys.86.153}{Rev. Mod. Phys.
  {\bf 86}, 153--185}~(2014).

\bibitem{ebadi2021quantum}
Sepehr Ebadi, Tout~T Wang, Harry Levine, Alexander Keesling, Giulia Semeghini,
  Ahmed Omran, Dolev Bluvstein, Rhine Samajdar, Hannes Pichler, Wen~Wei Ho,
  et~al.
\newblock ``{Quantum phases of matter on a 256-atom programmable quantum
  simulator}''.
\newblock \href{https://dx.doi.org/10.1038/s41586-021-03582-4}{Nature {\bf
  595}, 227--232}~(2021).

\bibitem{Scholl2022xxz}
P.~Scholl, H.~J. Williams, G.~Bornet, F.~Wallner, D.~Barredo, L.~Henriet,
  A.~Signoles, C.~Hainaut, T.~Franz, S.~Geier, A.~Tebben, A.~Salzinger,
  G.~Z\"urn, T.~Lahaye, M.~Weidem\"uller, and A.~Browaeys.
\newblock ``{Microwave Engineering of Programmable $XXZ$ Hamiltonians in Arrays
  of Rydberg Atoms}''.
\newblock \href{https://dx.doi.org/10.1103/PRXQuantum.3.020303}{PRX Quantum
  {\bf 3}, 020303}~(2022).

\bibitem{ibrahim2022pulse}
Mohannad Ibrahim, Hamed Mohammadbagherpoor, Cynthia Rios, Nicholas~T Bronn, and
  Gregory~T Byrd.
\newblock ``{Pulse-Level Optimization of Parameterized Quantum Circuits for
  Variational Quantum Algorithms}''~(2022).
\newblock  \href{http://arxiv.org/abs/2211.00350}{arXiv:2211.00350}.

\bibitem{meitei2020gate}
Oinam~Romesh Meitei, Bryan~T. Gard, George~S. Barron, David~P. Pappas,
  Sophia~E. Economou, Edwin Barnes, and Nicholas~J. Mayhall.
\newblock ``Gate-free state preparation for fast variational quantum
  eigensolver simulations''.
\newblock \href{https://dx.doi.org/10.1038/s41534-021-00493-0}{npj Quantum
  Information {\bf 7}, 155}~(2021).

\bibitem{mcclean2018barren}
Jarrod~R McClean, Sergio Boixo, Vadim~N Smelyanskiy, Ryan Babbush, and Hartmut
  Neven.
\newblock ``Barren plateaus in quantum neural network training landscapes''.
\newblock \href{https://dx.doi.org/10.1038/s41467-018-07090-4}{Nature
  communications {\bf 9}, 1--6}~(2018).

\bibitem{grant2019initialization}
Edward Grant, Leonard Wossnig, Mateusz Ostaszewski, and Marcello Benedetti.
\newblock ``{An initialization strategy for addressing barren plateaus in
  parametrized quantum circuits}''.
\newblock \href{https://dx.doi.org/doi.org/10.48550/arXiv.1903.05076}{Quantum
  {\bf 3}, 214}~(2019).

\bibitem{skolik2021layerwise}
Andrea Skolik, Jarrod~R McClean, Masoud Mohseni, Patrick van~der Smagt, and
  Martin Leib.
\newblock ``{Layerwise learning for quantum neural networks}''.
\newblock \href{https://dx.doi.org/doi.org/10.1007/s42484-020-00036-4}{Quantum
  Machine Intelligence {\bf 3}, 1--11}~(2021).

\bibitem{Achilles2022bch}
R{\"u}diger Achilles and Andrea Bonfiglioli.
\newblock ``{The early proofs of the theorem of Campbell, Baker, Hausdorff, and
  Dynkin}''.
\newblock \href{https://dx.doi.org/10.1007/s00407-012-0095-8}{Archive for
  History of Exact Sciences {\bf 66}, 295--358}~(2012).

\bibitem{lezcano2019cheap}
Mario Lezcano-Casado and David Mart{\i}nez-Rubio.
\newblock ``{Cheap orthogonal constraints in neural networks: A simple
  parametrization of the orthogonal and unitary group}''.
\newblock In International Conference on Machine Learning.
\newblock \href{https://dx.doi.org/10.48550/arXiv.1901.08428}{Pages
  3794--3803}.
\newblock PMLR~(2019).

\bibitem{mari2021estimating}
Andrea Mari, Thomas~R. Bromley, and Nathan Killoran.
\newblock ``{Estimating the gradient and higher-order derivatives on quantum
  hardware}''.
\newblock \href{https://dx.doi.org/10.1103/PhysRevA.103.012405}{Phys. Rev. A
  {\bf 103}, 012405}~(2021).

\bibitem{Russell2013Geometric}
Benjamin Russell and Susan Stepney.
\newblock ``{Geometric {Methods} for {Analysing} {Quantum} {Speed} {Limits}:
  {Time}-{Dependent} {Controlled} {Quantum} {Systems} with {Constrained}
  {Control} {Functions}}''.
\newblock In Giancarlo Mauri, Alberto Dennunzio, Luca Manzoni, and Antonio~E.
  Porreca, editors, {Unconventional {Computation} and {Natural} {Computation}}.
\newblock \href{https://dx.doi.org/10.1007/978-3-642-39074-6_19}{Pages
  198--208}.
\newblock Lecture {Notes} in {Computer} {Science}Berlin, Heidelberg~(2013).
  Springer.

\bibitem{arvanitogeorgos_introduction_2003}
Andreas Arvanitogeōrgos.
\newblock ``An introduction to {Lie} groups and the geometry of homogeneous
  spaces''.
\newblock \href{https://dx.doi.org/10.1090/stml/022}{Volume~22}.
\newblock American Mathematical Soc. ~(2003).

\bibitem{Helgason1978Differential}
S~Helgason.
\newblock ``Differential geometry, lie groups, and symmetric spaces''.
\newblock \href{https://dx.doi.org/10.1090/chel/341}{American Mathematical
  Soc.}~(1978).

\bibitem{humphreys2012introduction}
James~E Humphreys.
\newblock ``Introduction to {L}ie algebras and representation theory''.
\newblock \href{https://dx.doi.org/10.1007/978-1-4612-6398-2}{Volume~9}.
\newblock Springer Science \& Business Media. ~(2012).

\end{thebibliography}

\clearpage
\onecolumngrid
\onecolumn
\appendix

\renewcommand{\thesection}{\Alph{section}}
\renewcommand{\thefigure}{\Alph{section}\arabic{figure}}
\renewcommand{\thesubsection}{\arabic{subsection}}
\renewcommand{\theequation}{\Alph{section}\arabic{equation}}
\counterwithin*{equation}{section}
\counterwithin*{figure}{section}
\setcounter{section}{0}
\setcounter{subsection}{0}

\setcounter{equation}{0}
\setcounter{figure}{0}
\setcounter{table}{0}
\setcounter{page}{1}
\makeatletter

\section{\texorpdfstring{$\SU{N}$ }{sun} and its Lie algebra \label{app:sun}}
Consider the special unitary Lie group $\SU{N}$:
\begin{align}
    \SU{N} = \{ X \in \mathbb{C}^{N\times N} | X^\dag X = I,\, \det{X}=1\}.
\end{align}
In general, quantum gates belong to the unitary group $\mathrm{U}(N)$, which drops the determinant 1 condition and thus allows for an additional global phase. Restricting ourselves to $\SU{N}$ therefore is physically equivalent.
Consider a curve $X(t):\R \to \SU{N}$, $t \in (-1,1)$ such that $X(0) = I$ and $d/dt X(t) |_{t=0}\equiv \dot X(0) = \Omega$. If we differentiate the unitarity condition with respect to $t$, we obtain
\begin{align}
    \frac{d}{dt} (X^\dag(t)X(t)) \big|_{t=0}&= 0,\\
    \dot{X}^\dag(t)X(t) + X^\dag(t)\dot{X}(t) |_{t=0} &= 0\\
    \Omega^\dag + \Omega &=0,
\end{align}
and so we find that $\Omega^\dag = - \Omega$, i.e., $\Omega$ is a skew-Hermitian matrix. The matrices $\Omega$ from all such curves are elements of the Lie algebra 
\begin{align}
    \su{N} =\{\Omega\in\mathbb{C}^{N\times N} | \Omega^\dag=-\Omega,\, \Tr{\Omega}=0\}.
\end{align}
The second condition, $\Tr{\Omega}=0$, can be found by realizing that the Lie algebra is connected to a Lie group via the exponential map:
$e^{tX}\in \SU{N}$, $\forall t\in\mathbb{R}$ and $X\in\su{N}$ and so $1 = \det{e^{X}} = e^{\Tr{X}}$. Instead of considering a curve going through the identity at $t=0$, we can consider a curve going through a point $U$ with directional derivative $\Omega$. All curves $X(t)$ with directional derivative $\Omega$ at the point $U$ form an equivalence class. The set of these equivalence classes forms a vector space called the tangent space $T_{U} \,\SU{N}$ at the point $U$, and in particular we may identify $T_{I}\SU{N}=\su{N}$.
The group product $L_U$, i.e., left multiplication by an element $U$, is a local homeomorphism, which induces a linear map between two tangent spaces:  $d(L_U)\: |_V: T_{V} \,\SU{N} \to T_{L_U(V)}\, \SU{N}$  where $U,V\in \SU{N}$. $d(L_U)$ is called the differential of $L_U$. If we take $V = I$ to be the identity, then we see that we can move from the tangent space at the identity to the tangent space at any point in $\SU{N}$ by left multiplication of the group
$d(L_U)|_I:  \su{N} \to T_{U} \SU{N}$. Hence the tangent space at $U$ is given by
\begin{align}
    T_U\, \SU{N} = \{U \Omega |\,\Omega \in \su{N}\}.
\end{align}
The above derivation can provide us with a different understanding of a what the Lie algebra is. Starting from the group $\SU{N}$, we can consider all left-invariant vector fields $X$ on $\SU{N}$, 
\begin{align}
    d(L_U)|_V(X_V) = X |_{U\cdot V} \label{eq}
\end{align}
where $X\in \mathfrak{X}(\SU{N})$, which is the space of vector fields on $\SU{N}$. The left-invariant vector fields form a vector and since they are closed under the Lie bracket $[X,Y]$, they define a Lie algebra.

\section{The exponential map and its gradient \label{app:BCH}}
The following is due to \cite[Chapter 1, Theorem 5]{rossmann2002lie}. Let $G$ be a matrix Lie group $G \in \GL{N,\mathbb{C}}$ with a corresponding Lie algebra $\mathfrak{g}$. Define conjugation by $h\in G$ to be the transformation $c_h: G\to G$ given by $c_h(g) = h g h^{-1}$. Note that $c$ is an (inner) automorphism of $G$, since it is a isomorphism from $G$ onto itself. Let $\exp: \mathfrak{g} \to G$ be the exponential map from the Lie algebra to the group. Taking $X\in \mathfrak{g}$ and $t\in\mathbb{R}$, the differential of the conjugation map at the identity is
\begin{align}
    d(c_h)(X)=\frac{d}{dt} (h e^{t X} h^{-1}) \bigg|_{t=0} = h X h^{-1},
\end{align}
which maps an element of the Lie algebra to another element of the Lie algebra.
This map is called the adjoint representation $\Ad{h}: \mathfrak{g}\to\mathfrak{g}$ and takes $X\mapsto h X h^{-1}$. Given $X,Y\in\mathfrak{g}$, we compute
\begin{align}
    \frac{d}{dt}\Ad{e^{t X}} Y \bigg|_{t=0} &= \frac{d}{dt}e^{t X} Y e^{-t X}\bigg|_{t=0}\\
    &= XY - YX\\
    &= \comm{X}{Y}\\
    & = \ad{X}Y.
\end{align}
The operator $\ad{\hspace{-1.5mm}}: \mathfrak{g} \times \mathfrak{g} \to \mathfrak{g}$ is the Lie bracket on $\mathfrak{g}$, which for our purposes will be the standard commutator.
It now follows that
\begin{align}
    \frac{d}{dt}\Ad{e^{tX}} Y &= \frac{d}{dt}\left(e^{t X} Y e^{-tX}\right)\\
    &= X e^{tX} Y e^{-tX} + e^{tX} Y e^{-tX} (-X)\\
    & = \ad{X}(\Ad{e^{t X}}Y).
\end{align}
With the boundary condition $\Ad{e^{t X}}|_{t=0}=\mathrm{Id}$ we find that the above differential equation is solved by
\begin{align}
    \Ad{e^X} &= e^{\ad{X}}\label{eq:Ad_ad}
\end{align}
at $t=1$. Consider now the following parameterized matrix function $Y:\mathbb{R}
\times \mathbb{R}\to G$,
\begin{align}
    Y(s,t) = e^{-sX(t)} \frac{\partial}{\partial t} e^{sX(t)},
\end{align}
where $X(t)$ is a curve on $\mathfrak{g}$. We then find
\begin{align}
    \frac{\partial Y(s,t)}{\partial s} &= e^{-s X(t)}(-X(t)) \frac{\partial }{\partial t} e^{s X(t)} + e^{-s X(t)} \frac{\partial }{\partial t} (X(t) e^{s X(t)})\\
    &= -e^{-s X(t)}X(t) \frac{\partial }{\partial t} e^{s X(t)} + e^{-s X(t)}  X(t)  \frac{\partial }{\partial t}e^{s X(t)} + e^{-s X(t)} \frac{d X(t)}{d t} e^{s X(t)}\\
    &= e^{-s X(t)} \frac{d X(t)}{d t} e^{s X(t)}\\
    & = \Ad{e^{-s X(t)}} \frac{d X(t)}{d t}.
\end{align}
Then with equation \Cref{eq:Ad_ad}, we find
\begin{align}
    \frac{\partial Y(s,t)}{\partial s} = e^{-\ad{s X(t)}} \frac{d X(t)}{d t}. 
\end{align}
Using that $Y(0,t) = 0$, we find by integration
\begin{align}
    Y(1,t) &= \int_{0}^1 ds\,\frac{\partial Y(s,t)}{\partial s}.
\end{align}
Estimating the above integral forms the basis of the stochastic parameter-shift rule (see App. \ref{app:sps}). Continuing, 
\begin{align}
    Y(1,t) &=  \int_{0}^1 ds \sum_{n=0}^{\infty} \frac{(-1)^n s^n}{n!} (\ad{X})^n\frac{d X(t)}{d t}\\
    & = \left[\sum_{n=0}^{\infty} \frac{(-1)^n s^{n+1}}{(n+1)!} (\ad{X})^n\right]_{s=0}^{s=1}\frac{d X(t)}{d t}\\
    &= \left(\sum_{n=0}^{\infty} \frac{(-1)^n }{(n+1)!} (\ad{X})^n\right)\frac{d X(t)}{d t}.
\end{align}
Hence we see that 
\begin{align}
    \frac{d}{dt} e^{X(t)} 
    =e^{X(t)} Y(1,t) 
    =e^{X(t)} \left(\sum_{n=0}^{\infty} \frac{(-1)^n }{(n+1)!} (\ad{X})^n\right)\frac{d X(t)}{d t} ,\label{eq:omega_lie}
\end{align}
which gives \Cref{eq:adj}. 

Note that at this point the Baker-Campbell-Hausdorff formula can be derived with \Cref{eq:omega_lie} by considering the derivative of
\begin{align}
    e^{Z(t)} = e^{tX} e^{tY} ,
\end{align}
and subsequent integration of the derivative of $Z(t)$~\cite{Achilles2022bch}.

\section{Connection between Riemannian and standard gradient flow on \texorpdfstring{$\SU{N}$ }{sun}\label{app:riemann}}
Here, we highlight the connection between a Riemannian gradient flow on $\SU{N}$ with respect and \Cref{eq:circuit_with_omega}~\cite{SchulteHerbruggen2010gradflow, Wiersema2022riemann}. We take the Hilbert-Schmidt inner product $\langle A,B \rangle = \Tr{A^\dag B}$ as a Riemannian metric on $\SU{N}$ and write the cost as 
\begin{align}
    C(U) = \Tr{U\rho U^\dag H}.
\end{align}
The differential of the cost $d C(U): T_U \SU{N} \to \mathbb{R}$, evaluated at the point $U\Omega$, where $\Omega\in\su{N}$ and $U\in \SU{N}$, is then
\begin{align}
    d C(U) (U \Omega) &= d(\Tr)(U\rho U^\dag H)\circ d (U\rho U^\dag H) (U \Omega)\label{eq:dCU}\\
    &= \Tr{d(U)\rho U^\dag H + U\rho d(U^\dag) H}(U\Omega)\\
    &= \Tr{U\Omega\rho U^\dag H + U\rho \Omega^\dag U^\dag H}\\
    &= \Tr{\rho U^\dag H U\Omega - U^\dag H U \rho \Omega}\\
    &= \Tr{\comm{\rho}{U^\dag H U} \Omega}\\
    &= \langle -U\comm{\rho}{U^\dag H U} , U\Omega \rangle
\end{align}
via the chain rule and using that $\langle A, B \rangle = \langle UA, UB \rangle$. The compatibility condition for the Riemannian gradient tells us that 
\begin{align}
    d C(U) (U \Omega) = \langle \grad C(U), U\Omega\rangle.
\end{align}
Hence we can identify
\begin{align}
    \grad C(U) = -U\comm{\rho}{U^\dag H U}
\end{align}
with the corresponding gradient flow
\begin{align}
    \dot{U} = \grad C(U).
\end{align}
Next, we follow the results of~\cite{lezcano2019cheap} in our notation. Consider the cost function
\begin{align}
    C(X) = \Tr{e^X\rho e^{-X} H}
\end{align}
where $X\in \su{N}$. Note that although the minimum of the function is unchanged, the parameterization of a unitary via the Lie algebra changes the resulting gradient flow. To see this, we consider again the differential,
\begin{align}
    d (C\circ\exp X) ( A) &= d(\Tr)(e^X\rho e^{-X} H)\circ d (e^X\rho e^{-X} H) (A)\\
    &= \Tr{d (e^X)\rho U^\dag H + U\rho d(e^{-X}) H}(A),
\end{align}
where now $d (C\circ\exp X):\su{N}\to\mathbb{R}$ and $A\in\su{N}$.
We can now make use of the result in \Cref{eq:omega_lie},
\begin{align}
    d(e^X) (Y)
    &=e^{X} \left(\sum_{n=0}^{\infty} \frac{(-1)^n }{(n+1)!} (\ad{X})^n\right) (Y)\\
    &=e^X \Phi_X (Y)
\end{align}
to obtain 
\begin{align}
    d (C\circ\exp X) ( A)  = \Tr{e^X \Phi_X(A) \rho e^{-X} H + e^X\rho \Phi_X^\dag(A) e^{-X}  H}.
\end{align}
Using that $\Phi_X^\dag(A) = - \Phi_X(A)$ and $\Phi_X^\dag(\Phi_X(A)) = A$, we then have
\begin{align}
    d (C\circ\exp X) ( A)& = \Tr{\comm{\rho}{e^{-X} H e^X} \Phi_X(A)}\\
    &= \langle -\comm{\rho}{e^{-X} H e^X},\Phi_X(A)\rangle\\
    &=\langle \Phi_X(\comm{\rho}{e^{-X} H e^X}),A\rangle.
\end{align}
Hence the gradient on $\su{N}$ is 
\begin{align}
    \grad C(X) = \Phi_X(\comm{\rho}{e^{-X} H e^X}).
\end{align}
We therefore see that only when $\Phi_X = I$ do we obtain the Riemannian gradient on $\SU{N}$; hence only if $X\in\mathfrak{g}$ where $\mathfrak{g}$ is abelian. The optimization path followed by optimizing the parameters of an $\mathbb{SU}(N)$ gate is thus different from the one following a Riemannian gradient descent on $\SU{N}$.

\section{The generalized parameter-shift rule} \label{app:gps}
When using our $\mathbb{SU}(N)$ gate in an application that involves gradient-based optimization, like demonstrated in the numerical experiments in this work, we require calculating the partial derivatives in \Cref{eq:circuit_with_omega}. Here we provide the details for how this can be achieved in practice via the generalized parameter-shift rule (GPSR)~\cite{Wierichs2022grads, Izmaylov2021grads, Oleksandr2021grads}.
Without loss of generality, we rewrite the cost function in \Cref{eq:cost} as
\begin{align}
    C(t) = \Tr{\tilde{H} e^{t \Omega}\tilde{\rho} e^{-t \Omega}} \label{eq:cx},
\end{align}
where we absorbed the rest of the circuit into $\tilde{\rho}$, $\tilde{H}$ and fixed any other parameters in the circuit. Computing the derivative of \Cref{eq:cx} with respect to $t$ is equivalent to the problem of finding the gradient in \Cref{eq:circuit_with_omega} at $t=0$. For the numerical experiments in this paper we make use of the particular implementation of the GPSR in~\cite{Oleksandr2021grads} as well as the alternative method outlined in App. \ref{app:decompose_generators}.

The skew-Hermitian operator $\Omega$ in \Cref{eq:cx} has (possibly degenerate) eigenvalues $\{i\lambda_j \}$. We define the set of unique spectral gaps as $\Gamma =\{\abs{\lambda_j - \lambda_{j'}}\}$ where $j'>j$. Note that for $d$ distinct eigenvalues, the number of unique spectral gaps $R$ is bounded via $R\leq d(d-1)/2$. We relabel every unique spectral gap with an integer, i.e.~we write $\Delta_r \in \Gamma$, and define the corresponding vector $\bm{\Delta} = (\Delta_1,\ldots , \Delta_R)$.
We pick a set of parameter shifts that are equidistant and create a vector of $R$ shifts $\bm{\delta} = (\delta_1,\ldots, \delta_R)$ where
\begin{align}
    \delta_n = \frac{(2n - 1)\pi}{4R}, \quad n=1,\ldots, R.
\end{align}
Next, we create the length $R$ cost vector $\bm{c}$ and the $R\times R$ matrix $\mathrm{M}$
\begin{align}
    \bm{c} =   
    \begin{pmatrix}
    C(\delta_1) - C(-\delta_1)\\
    C(\delta_2) - C(-\delta_2)\\
    \vdots\\
    C(\delta_R) - C(-\delta_R)
    \end{pmatrix}, \quad
    \mathrm{M}(\bm{\delta}) =   
    \begin{pmatrix}
    2\sin(\delta_1 \Delta_1) & \hdots & 2\sin(\delta_1 \Delta_R)\\
    2\sin(\delta_2 \Delta_1) & \hdots & 2\sin(\delta_2 \Delta_R)\\
    \vdots & \vdots  &\vdots \\
    2\sin(\delta_R \Delta_1) & \hdots & 2\sin(\delta_R \Delta_R)
    \end{pmatrix}.
\end{align}
We then calculate the coefficient vector as
\begin{align}
    \bm{r} = (\mathrm{M}(\bm{\delta}))^{-1} \cdot \bm{c},
\end{align}
which finally gives the gradient
\begin{align}
    \frac{d C(x)}{dt}\bigg|_{t=0} = \bm{\Delta} \cdot \bm{r}.
\end{align}
Since the final gradient is exact, finite shot estimates of all $c(\delta_n)$'s will produce an unbiased estimate of $dC(x) / dt$,
\begin{align}
    \frac{d C(x)}{dt}\bigg|_{t=0} = \bm{\Delta} \cdot(\mathrm{M}(\bm{\delta}))^{-1} \cdot\mathbb{E}\left[ \bm{c}\right],
\end{align} 
where we pulled out $\bm{\Delta}$ and $ (\mathrm{M}(\bm{\delta}))^{-1}$ since they are constant.
The difficulty of obtaining an accurate estimate of the gradient is determined by the variance of this estimator, which is given by
\begin{align}
    \mathrm{Var}\left[\frac{d C(x)}{dt}\bigg|_{t=0}\right] = (\bm{\Delta} \cdot(\mathrm{M}(\bm{\delta}))^{-1})^{\odot 2} \cdot\left(\mathbb{E}\left[ \bm{c}^{\odot 2}\right] - \mathbb{E}\left[ \bm{c}\right]^{\odot 2} \right),
\end{align}
where we used $\odot 2$ to emphasize that the squares are taken elementwise. We assume that the estimates for each shifted circuit obey normal statistics and so since these are independent, we can write
\begin{align}
    \mathbb{E}\left[ \bm{c}^{\odot 2}\right] - \mathbb{E}\left[ \bm{c}\right]^{\odot 2} \approx \frac{1}{N_{\mathrm{shots}}}
    \begin{pmatrix}
    \sigma^2(\delta_1) + \sigma^2(-\delta_1)\\
    \sigma^2(\delta_2) + \sigma^2(-\delta_2)\\
    \vdots\\
    \sigma^2(\delta_R) + \sigma^2(-\delta_R)
    \end{pmatrix},
\end{align}
where $\sigma^2(\pm\delta_n)$ is the variance of the cost for each shifted circuit. If we assume that the dependence of $\sigma$ on the shifts is mild, i.e., $\sigma(\delta) \approx \sigma_0$ then the total variance will only depend on the prefactor. Setting the estimate $\mathbb{E}\left[ \bm{c}^{\odot 2}\right] - \mathbb{E}\left[ \bm{c}\right]^{\odot 2} = (\sigma_0^2,\sigma_0^2,\ldots,\sigma_0^2)$ then finally gives
\begin{align}
    \mathrm{Var}\left[\frac{d C(x)}{dt}\bigg|_{t=0}\right] &\approx 2 \sigma_0^2\left(\sum_n \Delta_n \mathrm{M}_{nm}^{-1}(\bm{\delta})\right)^{2}.
\end{align}
One can minimize this quantity with respect to $\bm{\delta}$ to find the optimal set of shifts for the gradient estimation~\cite{Oleksandr2021grads}.

\section{Alternative differentiation of \texorpdfstring{$\SU{N}$}{SU(N)} gates}\label{app:alt_diff}
In this section we summarize a number of alternative differentiation techniques that may be applied to the presented $\mathbb{SU}(N)$ gates.
In particular, we discuss the stochastic parameter-shift rule, which was created for multi-parameter gates, finite differences as a standard tool in numerical differentiation, as well as an alternative to the general parameter-shift rule above which also exploits the notion of effective generators.

\subsection{The Stochastic parameter-shift rule}\label{app:sps}
The stochastic parameter-shift rule~\cite{Crooks2019grads} relies on the following operator identity~\cite{Wilcox1967exponential}
\begin{align}
    \frac{\partial e^{Z(x)}}{\partial x} = \int_{0}^1 ds\, e^{sZ(x)} \frac{\partial Z(x)}{\partial x}e^{(1-s)Z(x)},
\end{align}
for any bounded operator $Z(x)$. We now fix all parameters $\theta_m$ for $m\neq l$ and rewrite the cost in \Cref{eq:cost} as
\begin{align}
    c(x) = \Tr{H e^{i(x G_l + A')}\rho e^{-i(x G_l + A')}},\quad A' \equiv \sum_{m\neq l} \theta_m G_m . \label{eq:csps}
\end{align}
Then, if we take $Z(x)$ to be the operator $Z(x) = i(x G_l + A')$ we can construct the gradient of \Cref{eq:csps} as
\begin{align}
    \frac{\partial c(x)}{\partial x} = \int_{0}^1 ds \left(C_+(x,s) - C_-(x,s)\right) ,\label{eq:grad_sps}
\end{align}
where 
\begin{align}
    C_\pm(x,s) &= \Tr{H V_\pm(x) \rho V^\dag_\pm(x)}\\
    V_\pm(x) &= e^{is(x G_l + A')}e^{\pm \frac{\pi}{4} G_l} e^{i(1-s)(x G_l + A')}.
\end{align}
Hence, similar to our method, the gradient evaluation requires adding gates to the circuit and evaluating the new circuit. However, the stochastic parameter-shift rule comes at a significant cost: the evaluation of the integral in \Cref{eq:grad_sps}. In practice, one approximates this integral by sampling values of $s$ uniformly in the interval $(0,1)$ and then calculating the costs $C_\pm(x,s)$ with a finite-shot estimate. Although this produces an unbiased estimator, we find that the variance of this estimator is larger than ours, see \Cref{fig:sampled_grad}.
In addition, this method leads to a bigger number of unique circuits to compute the derivative, increasing the compilation overhead for both hardware and simulator implementations.

\subsection{Finite differences}\label{app:fd}
Finite differences are widely used to differentiate functions numerically. We briefly discuss this method in the context of variational quantum computation (VQC) and refer the reader to recent works comparing and optimizing differentiation techniques for VQC~\cite{mari2021estimating, bittel2022fast}.

In particular, we consider the central difference recipe
\begin{align}\label{eq:central_difference}
    \partial_{\text{FD}, \theta_j} C(\btheta) = \frac{1}{\delta} \left[C\left(\btheta+\frac\delta2 \bm{e}_j\right) - C\left(\btheta-\frac\delta2 \bm{e}_j\right)\right],
\end{align}
where $\delta$ is a freely chosen shift parameter and $\bm{e}_j$ is the $j$th canonical basis vector.
This recipe is an approximation of $\partial_{\theta_j}C(\btheta)$, making the corresponding estimator on a shot-based quantum computer biased.
This bias, which depends on $\delta$, has to be traded off against the variance of the estimator, which grows approximately with $\delta^{-2}$.

In classical computations, the numerical precision cutoff plays the role of the variance. Due to the high precision in classical computers, this leads to optimal shifts $\delta\ll1$, which allows treating the bias to leading order in $\delta$ and thus enables rough estimates of the optimal $\delta^\ast$ in advance. On a quantum computer, however, the variance typically is more than ten orders of magnitude larger, leading to a very different $\delta^\ast$, which furthermore depends on the function and derivative values.
As a consequence, shifts of $\mathcal{O}(1)$ become a reasonable choice, highlighting the similarity of the central difference recipe to the two-term parameter-shift rule~\cite{mari2021estimating}.

As a demonstration of the above, and in preparation for the numerical experiments shown in \Cref{fig:exact_grad,fig:sampled_grad}, we compute the central difference gradient for a random single-qubit Hamiltonian, a single $\mathbb{SU}(2)$ gate $U(\btheta)=\exp(iaX+ibY)$ and $\delta\in\{0.5, 0.75, 1.0\}$.
For this, we evaluate the mean and standard error $50$ times and show the difference to the exact derivative in \Cref{fig:finite_diff_tune}.
As expected, we observe that the bias increases with $\delta$ and that the variance is suppressed with larger $\delta$. We determine $\delta=0.75$ to be a reasonable choice for the purpose of the demonstration in \Cref{fig:exact_grad,fig:sampled_grad}, but stress that for any other circuit, qubit count, Hamiltonian, and even for a different parameter position $\btheta$ for this circuit, the optimal shift size needs to be determined anew.

\begin{figure}[htb!]
    \centering
    \includegraphics[width=\textwidth]{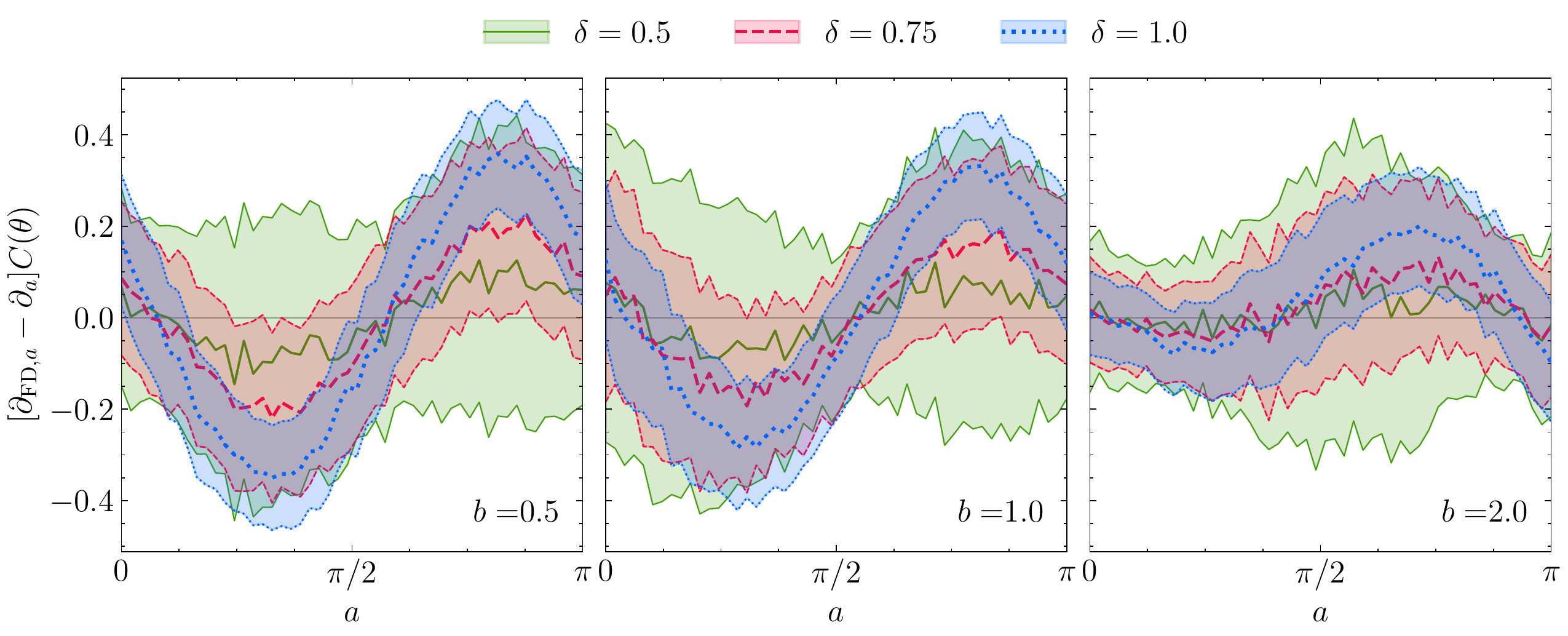}
    \caption{
    Error of the central difference gradients with $\delta=0.5,0.75,1.0$ for the single-qubit example from \Cref{fig:exact_grad,fig:sampled_grad}. The value of the second parameter again is fixed to $b=0.5,1.0,2.0$ in the panels (from left to right).
    The shift parameter $\delta$ influences the strengths of bias and variance, leading to a trade-off. For smaller $\delta$, the variance is enhanced due to the coefficients in \Cref{eq:central_difference} that scale with $\delta^{-1}$. For larger $\delta$, the bias based on the approximate nature of \Cref{eq:central_difference} is increased. We find $\delta=0.75$ to be a reasonable choice for this particular circuit, Hamiltonian and parameter position $\btheta$.
    }
    \label{fig:finite_diff_tune}
\end{figure}

\subsection{Decomposing effective generators for differentiation}\label{app:decompose_generators}

In \Cref{alg:grad} we suggest to use the generalized parameter-shift rule~\cite{Izmaylov2021grads, Wierichs2022grads, Oleksandr2021grads} in order to compute the partial derivatives $\frac{\partial}{\partial \theta_l}C(\btheta)$ independently.
In addition, \Cref{th:spectral} bounds the number of frequencies occurring in the univariate auxiliary cost function $C(t) = \Tr{U(\btheta)e^{t\Omega_l(\btheta)}\rho e^{-t\Omega_l(\btheta)}U^\dag(\btheta)H}$ and the corresponding number of parameter shifts required during differentiation.

Realizing the shift rule requires us to implement not only $U(\btheta)$---which is necessary to compute $C(\btheta)$ itself---but also the gate $e^{t_r \Omega_l(\btheta)}$ for $2R$ shift values $t_r$ and each $\Omega_l$ separately.
Alternatively, we may follow the approach to decompose all effective generators $\Omega_l$ and compute the derivative as a linear combination of the derivatives for simpler auxiliary gates, similar to~\cite{Izmaylov2021grads}. In particular, we again choose the Pauli basis of $\su{N}$ for this decomposition.

Decompose the effective generators $\Omega_l(\btheta)$ as
\begin{align}
    \Omega_l(\btheta) &= \sum_{m} \omega_{lm}(\btheta) G_m, \qquad 
    \omega_{lm}(\btheta) = \frac{1}{N} \Tr{G_m \Omega_l(\btheta)}.
\end{align}
Note that the coefficients are purely imaginary due to the skew-Hermiticity of $\Omega_l(\btheta)$.
The partial derivative we are interested in can then be written as 
\begin{align}
    \frac{\partial}{\partial \theta_l} C(\btheta)
    &= \Tr{H U(\btheta)\left[\sum_{m=1}^d \omega_{l m}(\btheta) G_m,\rho\right] U^\dag(\btheta)}\\
    &= \sum_m \omega_{l m}(\btheta) \Tr{H U(\btheta)[G_m,\rho] U^\dag(\btheta)}\\
    &= \sum_m \omega_{l m}(\btheta) 2i \frac{\mathrm{d}}{\mathrm{d}t}\Tr{H U(\btheta)\left[\exp\left\{-i\frac{t}{2} G_m\right\},\rho\right] U^\dag(\btheta)}\bigg|_{t=0}\\
    &= \sum_m \widetilde{\omega}_{l m}(\btheta)  \frac{\mathrm{d}}{\mathrm{d}t}C_{G_m}(\btheta, t)\big|_{t=0}.
\end{align}
Here we abbreviated $\widetilde{\omega}_{lm}(\btheta) = 2i\omega_{lm}(\btheta)$ and wrote $C_{G_m}(\btheta, t)$ for the cost function with a rotation gate with parameter $-t/2$ about $G_m$ inserted before $U(\btheta)$. This modified cost function can be differentiated with respect to $t$ using the original two-term parameter-shift rule, as the inserted gate is generated by (the multiple of) a Pauli string.

The above linear combination of Pauli rotation derivatives can be reused for all partial derivatives, so that the full gradient for one $\mathbb{SU}(N)$ gate is given by
\begin{align}
    \nabla C(\btheta) &= \widetilde{\omega}(\btheta) \cdot \bm{\mathrm{d}C},\\
    \bm{\mathrm{d}C} &= \left(\begin{matrix}
        \frac{\mathrm{d}}{\mathrm{d}t}C_{G_1}(\btheta, t)\big|_{t=0} \\
        \vdots\\
        \frac{\mathrm{d}}{\mathrm{d}t}C_{G_d}(\btheta, t)\big|_{t=0}
    \end{matrix}\right).
\end{align}

So far we did not discuss the number of Pauli strings occurring in the decomposition of the generators $\Omega_l$. As can be seen from \Cref{eq:omega} and our definition of the DLA in \Cref{sec:resource}, this number is bounded by the size of the DLA, and we again remark that this bound will be saturated for most values of $\btheta$.
As two shifts are required for each Pauli rotation, the gradient $\nabla C(\btheta)$ can thus be computed using $2\dim \mathcal{L}(A(\btheta))$ circuits, using Pauli rotations from the DLA and, e.g., shift angles $\pm\frac{\pi}{2}$.

As we only required a linear decomposition of $\Omega_l$, any other basis for the DLA may be used as well, potentially allowing for fewer shifted circuits or different inserted gates that may be more convenient to realize on hardware.

\section{\label{app:gate_speed_limit}Gate speed limit}
The following Lemmas are used in \Cref{sec:gate_speed_limits}.
\begin{lemma}
	\label{lemma:trace_of_square}
	For Hamiltonians of the form $H = \sum_m \theta_m G_m$, where $G_m$ are strings of $~\log_2 N$ Pauli operators, $\textrm{Tr}\left\{ H^2 \right\}= N \sum_m \theta_m^2$.
\end{lemma}
\begin{proof}
All Pauli strings $G_m$ are orthonormal with respect to the trace inner product, $\Tr(G_m^\dag G_n) = \delta_{n,m} N$. Using this gives 
\begin{align}
	\mathrm{Tr}\left\{ H^2 \right\} &= \sum_{m,n} \theta_m\theta_n \Tr{G_m G_n} \\ 
    &= N \sum_{m,n} \theta_m\theta_n \delta_{n,m}\\
    &= N \sum_m \theta_m^2.
\end{align}
\end{proof}

\begin{lemma}
	\label{lemma:path_distance}
    The length of a smooth curve on the Riemannian manifold $\SU{N}$, with metric $g(x,y) = \Tr{x^\dagger y}$, for a time-independent Hamiltonian $H$ after a fixed time $\tau$ only depends on the norm of $H$ and $\tau$ .
\end{lemma}
\begin{proof}
	The unitary evolution of $U(\bm{\theta}; t)$, parameterized by $t$, corresponds to a smooth curve on $\mathrm{SU}(N)$ with length according to the Riemannian metric, 
    $g(x,y) = \Tr{x^\dagger y}$~\cite{Russell2013Geometric}.
    Integrating over the metric norm through the tangent spaces from $t=0$ to final time $\tau$ gives the path length,
	\begin{align}
        L[U(\bm{\theta};t), \tau] &=  \int^{\tau}_{t=0} ds\\
		 &=  \int^{\tau}_{0} \sqrt{g(\dot{U}(\bm{\theta}; t), \dot{U}(\bm{\theta}; t))} dt\\
		 &= \int^{\tau}_{0} \sqrt{ \textrm{Tr}\left\{\dot{U}^\dagger(\bm{\theta}; t) \dot{U}(\bm{\theta}; t) \right\}} dt,
	\end{align}
	where $\dot{U} = \frac{dU}{dt}$. From Schrödinger evolution,
	\begin{equation}
		\frac{d U(\bm{\theta}; t)}{d t} = -i H U(\bm{\theta}; t),
	\end{equation}
	we find
	\begin{equation}
		L[U(\bm{\theta}; t), \tau] = \tau \sqrt{\textrm{Tr}\left\{ H^2 \right\}},
	\end{equation}
	for time-independent Hamiltonians. Therefore for all Hamiltonians with fixed norm $\Tr{H^2}$, the path distance travelled after time $\tau$ is the same regardless of the specific unitary evolution.
\end{proof}
\begin{lemma}
	\label{lemma:minimal_path}
    The minimal path between the identity element $I$ and a point $V$ on the Riemannian manifold of $\SU{N}$, with metric $g(x,y) = \Tr{x^\dagger y}$, is a geodesic curve $\gamma(t) = e^{X t}$ with $X \in \su{N}$.
\end{lemma}
\begin{proof}
    From Proposition~3.10~\cite{arvanitogeorgos_introduction_2003}, the geodesics of $\SU{N}$ are the one-parameter subgroups, given by $\gamma(t) = e^{X t}$.
    In general, multiple geodesics curves can give $V$ from the identity. The minimal path is the curve with the minimum length.
    Since geodesic curves are the extrema of the path length functional~\cite[Lemma~9.3]{Helgason1978Differential}, there must exist a geodesic, which is of the form $\gamma(t) = e^{X t}$, that is the minimal path to $V$.

\end{proof}

\subsection{\label{app:th:gate_speed_limit_proof}Proof of~\texorpdfstring{\Cref{th:gate_speed_limit}}{Theorem 1}}
We restate \Cref{th:gate_speed_limit} again for convenience.
\setcounter{theorem}{0}
\begin{theorem}
	For unitary gates generated by normalized time-independent Hamiltonians, consider a general circuit decomposition of two gates $U(\bphi^{(2)};t_2)U(\bphi^{(1)};t_1)$.  There exists an equivalent evolution with an $\mathbb{SU}(N)$ gate $U(\btheta;t_g) =U(\bphi^{(2)};t_2)U(\bphi^{(1)};t_1)$, with evolution time $t_g$, such that 
	\begin{align*}
		t_g \leq t_1+t_2,
	\end{align*}
	with equality if $\bphi^{(1)}+\bphi^{(1)} = \btheta$.
\end{theorem}
\begin{proof}
    The product $U(\bphi^{(2)};t_2)U(\bphi^{(1)};t_1)$ corresponds to a specific point $V$ on the manifold $\SU{N}$. By Lemma~\ref{lemma:minimal_path}, there exists a geodesic between the identity $I$ and $V$, given by the curve $e^{Xt}$ that is of minimal length. We can parameterize this geodesic as $U(\bphi^{(2)};t_g) = \exp{\Bar{A}(\btheta) t_g}$, which is always possible since $A(\btheta)$ parameterizes an arbitrary point in $\su{N}$ and is an $\mathbb{SU}(N)$ gate. By Lemma~\ref{lemma:path_distance}, the length of this path only depends on the norm of $\Bar{A}(\btheta)$, which is $1$, and on $t_g$, which gives
    \begin{align}
    t_g = L[U(\btheta;t),t_g)].
    \end{align}
    Since this path is minimal, we have
    \begin{align}
        t_g \leq t_1+t_2
    \end{align}
    with equality if $\bphi^{(1)}+\bphi^{(1)} = \btheta$.
\end{proof}

\subsubsection{\label{app:th:gate_speed_limit_su2}Special case of~\texorpdfstring{$\SU{2}$}{}}
In the following we give the additional time for decomposing an optimal $\mathbb{SU}(2)$ gate into two gates. By \emph{optimal}, we refer to the geodesic along the minimal path length curve -- see Lemma~\ref{lemma:minimal_path}.

We consider the optimal $\mathbb{SU}(2)$ gate $U(\btheta;t_g)$ with geodesic evolution time $t_g$ together with a decomposition $U(\bphi^{(2)};t_2)U(\bphi^{(1)};t_1)=U(\btheta;t_g)$.
The decomposed circuit is given by two unitary evolutions. Each individual evolution $U(\tilde{\bphi}^{(\nu)};t_\nu)$ is a $\textrm{U}(1)$ rotation such that only a single basis element is required. With two rotations, the overall evolution is an element of $\SU{2}$. The corresponding $\su{2}$ algebra is spanned by three basis elements---the three Pauli matrices for example. The two rotations can be represented as lying on a Bloch sphere. A unitary transformation, $K \in \SU{N}$, can therefore be applied to the evolution such that $U(\tilde{\bphi}^{(\nu)};t_\nu) = K U(\bphi^{(\nu)};t_\nu)K^\dagger$ for
    \begin{align}
        \tilde{\bphi}^{(\nu)} = 
        \begin{pmatrix}          \sin(\alpha_{\nu})\cos(\beta_{\nu})\\ \sin(\alpha_{\nu})\sin(\beta_{\nu})\\
        \cos(\alpha_{\nu}) 
        \end{pmatrix},
    \end{align}
    where $\nu = 1,2$, with $\alpha_\nu$ and $\beta_\nu$ parameterizing the rotations. By construction $\tilde{\bphi}^{(\nu)} \cdot \tilde{\bphi}^{(\nu)} =1$ is normalised for all parameters $\alpha_\nu$ and $\beta_\nu$. The same transformation $K$ defines $U(\tilde{\btheta};t_g) = K U(\btheta;t_g) K^\dagger$ and gives the same relationship $U(\tilde{\btheta};t_g) =U(\tilde{\bphi}^{(2)};t_2) U(\tilde{\bphi}^{(1)};t_1) $. This is straightfoward to show
    \begin{align}
        U(\tilde{\btheta};t_g) &= K U(\btheta;t_g) K^\dagger \\ 
        &= K U(\bphi^{(2)};t_2)  U(\bphi^{(1)};t_2)  K^\dagger \\
        &= K U(\bphi^{(2)};t_2) K^\dagger K  U(\bphi^{(1)};t_2)  K^\dagger \\
        &= U(\tilde{\bphi}^{(2)};t_2) U(\tilde{\bphi}^{(1)};t_1).
    \end{align}
    We have
    \begin{align}
        \bar{A}(\tilde{\bphi}^{(\nu)} ) = \sin(\alpha_{\nu})\cos(\beta_{\nu}) G_1 + \sin(\alpha_{\nu})\sin(\beta_{\nu}) G_2 + \cos(\alpha_{\nu}) G_3,
    \end{align}
    where we choose $G_1 = i X$, $G_2 = i Y$, and $G_3 = i  Z$. These basis elements of $\su{2}$ generate the group $\SU{2}$. We also define the basis vector $\bm{G} = (G_1, G_2, G_3)$. 
    Exponentiation therefore gives the closed-form expression
    \begin{align}
        \exp{\bar{A}(\tilde{\bphi}^{(\nu)})t_\nu} &= \exp{\left( \tilde{\bphi}^{(\nu)} \cdot \bm{G} \right) t_\nu} \\
        &=\cos(t_\nu) I^{\otimes N_\mathrm{qubits}} + \sin(t_\nu) \tilde{\bphi}^{(\nu)} \cdot \bm{G}.
    \end{align}
    By the group composition law of $\SU{2}$, the product of two exponentials in $\SU{2}$ also gives a closed-form expression,
    \begin{align}
        \exp{\bar{A}(\tilde{\bphi}^{(2)})t_2} \exp{\bar{A}(\tilde{\bphi}^{(1)})t_1} &= \left(\cos(t_2) I^{\otimes N_\mathrm{qubits}} + \sin(t_2) \tilde{\bphi}^{(2)} \cdot \bm{G}\right)\left(\cos(t_1) I^{\otimes N_\mathrm{qubits}} + \sin(t_1) \tilde{\bphi}^{(1)} \cdot \bm{G} \right).
    \end{align}
    Collecting terms gives
    \begin{multline}
        \exp{\bar{A}(\tilde{\bphi}^{(2)})t_2} \exp{\bar{A}(\tilde{\bphi}^{(1)})t_1} = \left( \cos(t_1) \cos(t_2) - \tilde{\bphi}^{(1)}\cdot\tilde{\bphi}^{(2)} \sin(t_1)\sin(t_2)\right) I^{\otimes N_\mathrm{qubits}} \\ + \left( \cos(t_2) \sin(t_1) \tilde{\bphi}^{(1)} +  \cos(t_1) \sin(t_2) \tilde{\bphi}^{(2)} + i \sin(t_1)\sin(t_2) \tilde{\bphi}^{(1)}\times\tilde{\bphi}^{(2)} \right) \cdot \bm{G}.\label{eq:phi_1_phi_2_expression}
    \end{multline}
    The total evolution is 
    \begin{align}
        \exp{\bar{A}(\tilde{\bphi}^{(2)})t_2} \exp{\bar{A}(\tilde{\bphi}^{(1)})t_1} &= \exp{\bar{A}(\tilde{\btheta})t_g} \\
        &= \cos(t_g) I^{\otimes N_\mathrm{qubits}} + \sin(t_g) \tilde{\btheta} \cdot \bm{G}. \label{eq:theta_g_expression}
    \end{align}
    By comparison of Eqs.~\eqref{eq:phi_1_phi_2_expression} and~\eqref{eq:theta_g_expression}, we find
    \begin{align}
        \tilde{\btheta} = \frac{1}{\sin(t_g)}\left( \cos(t_2) \sin(t_1) \tilde{\bphi}^{(1)} +  \cos(t_1) \sin(t_2) \tilde{\bphi}^{(2)} + i \sin(t_1)\sin(t_2) \tilde{\bphi}^{(1)}\times\tilde{\bphi}^{(2)} \right),
    \end{align}
    and
    \begin{align}
         \cos(t_g) = \cos(t_1)\cos(t_2) - \tilde{\bphi}^{(1)}\cdot\tilde{\bphi}^{(2)}\sin(t_1)\sin(t_2).
    \end{align} 
    The additional evolution time is $\Delta t = t_d - t_g$, with $t_d = t_1 + t_2$ the total decomposed unitary evolution time. Due to the invariance of the scalar product, $\tilde{\bphi}^{(1)}\cdot\tilde{\bphi}^{(2)}= \bphi^{(1)}\cdot\bphi^{(2)}$, the additional time $\Delta t=t_d-t_g$ required by the decomposition is then given by
\begin{align*}
    \Delta t = t_d - \arccos\big(\cos(t_1)\cos(t_2) - \bphi^{(1)}\cdot\bphi^{(2)}\sin(t_1)\sin(t_2)\big) \geq 0.
\end{align*}

\section{Unique spectral gaps of Dynamical Lie Algebras \label{app:th4}}
\subsection{Proof of \Cref{th:spectral}}
We restate \Cref{th:spectral} here for convenience.
\begin{theorem}
The number of unique spectral gaps $R$ of $\Omega_l(\btheta)$ is upper bounded by the number of roots $|\Phi|$ of any maximal semi-simple DLA,
\begin{align}
    R \leq |\Phi| /2.
\end{align}
\end{theorem}
In the following we set $\mathfrak{g}$ to be a semisimple Lie algebra. A subspace $\mathfrak{a} \subseteq \mathfrak{g}$ is called a subalgebra if it is closed under the Lie bracket, i.e., if $\comm{a_1}{a_2}\in\mathfrak{a}$, $\forall a_1,a_2 \in \mathfrak{a}$. 
Since $\mathfrak{g}$ is a semisimple Lie algebra, it always contains a subalgebra called a Cartan subalgebra~\cite{hall2013lie} (Chapter 7, Definition 7.10).
\begin{definition}
A Cartan subalgebra $\mathfrak{h}$ of $\mathfrak{g}$ is a subalgebra that satisfies the following conditions:
\begin{enumerate}
    \item For all $ h_1, h_2 \in\mathfrak{h}$, $\comm{h_1}{h_2} = 0$.
    \item For all $ x \in \mathfrak{g}$, if $\comm{h}{x}=0$ for all $h\in\mathfrak{h}$, then $x \in \mathfrak{h}$. 
\end{enumerate}
\end{definition}
The first condition tells us that $\mathfrak{h}$ is a commutative subalgebra of $\mathfrak{g}$, while the second condition says that $\mathfrak{h}$ is maximal, i.e., there is no larger commutative subalgebra. The first step in proving  \Cref{th:spectral} is to make use of the following result:
\begin{theorem}\label{app:th:root}
\cite[Chapter VI, Theorem 1]{Serre2000complex}. If $\mathfrak{g}$ is a semisimple Lie algebra, we can write $\mathfrak{g}$ as a direct sum of the root spaces $\mathfrak{g}_\alpha$:
\begin{align}
    \mathfrak{g} = \bigoplus_{\alpha}\mathfrak{g}_\alpha,
\end{align}
where
\begin{align}
    \mathfrak{g}_{\alpha\in\mathfrak{h}^*} = \{x \in \mathfrak{g}| \mathrm{ad}_{h}(x) = \alpha(h) x,\,\forall h \in\mathfrak{h}\},
\end{align}
and $\alpha \in\mathfrak{h}^*$ are functionals on $\mathfrak{h}$.
That is, a root space is a subspace of $\mathfrak{g}$ on which the action of the adjoint representation of $\mathfrak{h}$ is described by a functional (and scalar multiplication).
\end{theorem}
The above decomposition is called a root space decomposition, which is an essential tool in classifications of Lie algebras~\cite{dynkin1957american, Serre2000complex}.
Since 
\begin{align}
    \mathfrak{g}_0 = \{ x \in\mathfrak{g} | \ad{h}(x) = 0,\forall h \in\mathfrak{h}\},
\end{align}
we find that $\mathfrak{h} = \mathfrak{g}_0$ and hence
\begin{align}
    \mathfrak{g} = \mathfrak{h} \oplus \bigoplus_{\alpha\neq 0}\mathfrak{g}_\alpha.
\end{align}
We then immediately see that
\begin{align}
    \dim{\mathfrak{g}} = \dim{\mathfrak{h}} + \sum_{\alpha\neq 0} \dim{\mathfrak{g}_\alpha}.
\end{align}
We can thus relate the dimensionality of a Lie algebra to the dimensionality of its Cartan subalgebra and its weight spaces. 
The second step of the proof relies on identifying the unique spectral gaps of $\Omega_l(\btheta)$ with the weight spaces $\mathfrak{g}_\alpha$. To achieve this, we will construct the linear operator $\ad{h}$ and apply it to the eigenbasis of $\Omega_l(\btheta)$ to show that the maps $\alpha$ can be identified with the spectral gaps of $\Omega_l(\btheta)$.

Consider an element $\Omega \in \mathfrak{g}$,
where $\mathfrak{g}\subseteq \su{N}$ is a non-trivial subalgebra and $\mathfrak{h}$ is a Cartan subalgebra of $\mathfrak{g}$. Since $\mathfrak{h}$ is the Lie algebra of a maximally abelian group, we can represent elements of $\mathfrak{h}$ by diagonal matrices. Since $\Omega$ is skew-Hermitian, there exists a unitary $V\in\SU{N}$ that diagonalizes $\Omega$, i.e., $V^{\dag} \Omega V=h$ with $h\in\mathfrak{h}$. Here, $V$ is the matrix with columns equal to the eigenvectors $v_k$ of $\Omega$ with corresponding eigenvalues $\lambda_k$. We can thus always choose a basis for $\mathfrak{g}$ such that $\Omega$ is diagonal. If $\Omega$ is non-degenerate, then it must be full rank, and therefore an element of $\mathfrak{h}$. All Cartan algebras are equivalent up to conjugacy, hence we can choose the matrix $h$ to be the diagonal matrix containing the eigenvalues of $\Omega$ to represent the Cartan subalgebra $\mathfrak{h}$. We now take $E_{nm}$ to be the matrix with entries $(n,m)$ equal to 1 and all other entries to 0. 
Define the operator
\begin{align}
    e_{nm} = V^\dag E_{nm} V,
\end{align}
and apply $\ad{h}$ to it:
\begin{align}
    \ad{h}(e_{nm})
    &= he_{nm}-e_{nm}h\\
    &= V^\dag \Omega E_{nm} V - V^\dag E_{nm} \Omega V\\
    &= V^\dag h E_{nm} V - V^\dag E_{nm} h V\\
    &= (\lambda_n - \lambda_m) e_{nm}.
\end{align}
This means that $\ad{h}$ has the eigenvectors $e_{nm}$ with corresponding eigenvalues $\alpha_{nm}(h) = \lambda_n - \lambda_m$~\cite{humphreys2012introduction}, and so we have identified the eigenvalue differences with the roots of the Lie algebra. We define the set of all roots as
\begin{align}
    \Phi = \{ \lambda_n - \lambda_m, n\neq m=1,\ldots, N\}.
\end{align}
Since the dimensionality of each weight space is one \cite[Chapter VI, Theorem 2(a)]{Serre2000complex}, we can see that
\begin{align}
    \sum_{\alpha\neq 0} \dim{\mathfrak{g}_\alpha} = |\Phi|.
\end{align}
Therefore, 
\begin{align}
    \dim{\mathfrak{g}} = \dim{\mathfrak{h}} + \abs{\Phi}.
\end{align}
We now set $\mathfrak{g}=\mathcal{L}(A(\btheta))$. If we take the absolute value of the elements of $\Phi$, we can identify $R = |\Phi|/2$, where the factor $1/2$ is to account for double the counting of the spectral gaps. Since $\Omega$ can be degenerate in general, we obtain the inequality $R \leq |\Phi|/2$. 
With this, the proof of \Cref{th:spectral} is completed.

\subsection{Examples}
Here, we give several examples of maximal DLAs and their corresponding value of $|\Phi|/2$. Analogous to the main text, we choose the Pauli representation but these results should hold for any irreducible representation of $\su{N}$.
\begin{enumerate}
    \item $\su{2}$: For a 1-qubit system, there are no non-trivial subalgebras, hence we can only look at the full special unitary Lie algebra $\su{2}$. Any $A(\btheta)$ that consists of two Pauli operators will generate this algebra, e.g.,
    \begin{align}
        A(\btheta) = i(\theta_1X  + \theta_2 Y)
    \end{align}
    will give $\mathcal{L}(A(\btheta)) = \su{2}$. A Cartan subalgebra of $\su{2}$ is given by $\mathfrak{h} = \Span{Z}$. We therefore find that $\dim{\mathfrak{g}} = 3$ and $\dim{\mathfrak{h}} = 1$ and so $|\Phi| = 2$. Hence we have $R\leq 1$ and need $2R\leq 2$ shifts. This matches the result in~\cite{Schuld2019grads}, where the parameter-shift rule was generalized from single Pauli matrices to Hermitian operators with two unique eigenvalues.
    
    \item TFIM: A DLA that has been studied before~\cite{larocca2022diagnosing, Kokcu2021cartan} is the 1D transverse field Ising-Model (TFIM) Hamiltonian:
    \begin{align}
    A(\btheta) = i (\theta_{1}X \otimes I + \theta_{2}I\otimes X  + \theta_{3}Z \otimes Z ),
    \end{align}

    with $\mathcal{L}(A(\btheta)) = \Span{X \otimes I, I \otimes X, Y\otimes Y,Z\otimes Z, Z\otimes Y,Y\otimes Z}$. 
    We can take $\mathfrak{h} = \Span{X \otimes I, I \otimes X}$ as a Cartan subalgebra and so $\dim{\mathfrak{g}} = 6$ and $\dim{\mathfrak{h}} = 2$, which gives $|\Phi| = 4$. Hence we need (at most) 4 shifts to obtain the gradient of an operator in the DLA of the TFIM, which corresponds to $\mathfrak{so}(4)$.

    \item $\su{4}$: The full Lie algebra of $\su{4}$ is spanned by
    \begin{align}
        A(\btheta) = \sum_m \theta_m G_m,
    \end{align}
    where $G_m \in \mathcal{P}^{4}$ are the tensor products as defined in \Cref{eq:paulis}. A Cartan subalgebra of $\su{4}$ is given by $\mathfrak{h} = \{Z\otimes I, I\otimes Z, Z\otimes Z\}$.
    This means that $\dim{\mathfrak{g}} = 15$ and $\dim{\mathfrak{h}} = 3$, which gives $|\Phi|=12$. Hence we have $R=6$ and need 12 shifts to obtain the gradient for a general operator in $\su{4}$.
\end{enumerate}

In the above examples, we have only been concerned with the dimensionality of the root system. We could go one step further and look at the structure of the root systems. It turns out that there exists only a finite set of root systems, which leads to the classification of all semisimple Lie algebras (such a program was originally carried out by Dynkin~\cite{dynkin1957american} and is explained in most textbooks on Lie algebras~\cite{rossmann2002lie, humphreys2012introduction, hall2013lie}). This allows us to make the following observation about DLAs and the $\mathbb{SU}(N)$ gates in our work: there is a finite number of families of $\mathbb{SU}(N)$ gates for each $N$, given by the possible DLAs. Again, we emphasize that this is independent of the representation of the algebra. We summarize the above results together with the identification of the corresponding classical group in \Cref{tab:classif}~\cite{rossmann2002lie} (Chapter 3, Table 3.4).
\begin{table}[htb!]
    \centering
    \begin{tabular}{c|c|c|c|c}
        Name & $\dim(\mathfrak{g})$ & $\dim(\mathfrak{h})$ & $|\Phi|$& Classical group\\\hline
        $\su{2}$ & 3 & 1 & 2& $A_1$\\
        $\mathfrak{so}(4)$ & 6 & 2 & 4 &$A_1 \times A_1\cong D_2$\\
        $\su{4}$ & 15 & 3 & 12& $A_2$\\
    \end{tabular}
    \caption{Examples of DLAs and the size of the root spaces. Each root system $\Phi$ can be identified with a Lie algebra of one of the classical groups $A_n,B_n,C_n,D_n$. The classical group $D_2$ corresponds to $\mathrm{SO}(4)$, with the corresponding Lie algebra $\mathfrak{so}(4)$ which has dimension $N(N-1)/2$.}
    \label{tab:classif}
\end{table}

\end{document}